\documentclass{conm-p-l}

\usepackage[english]{babel}

\usepackage{amsmath,amssymb,amsthm,amsfonts}
\usepackage{exscale}

\usepackage{bbm}
\usepackage{lmodern}

\usepackage{hyperref}

\usepackage{slashed}

\usepackage{graphicx}

\usepackage[all]{xy}
\CompileMatrices

\usepackage{tensor}

\usepackage{ifthen}

\numberwithin{equation}{section}

\newcommand{\K}{
\mathbb{K}
}

\newcommand{\N}{
\mathbb{N}
}

\newcommand{\R}{
\mathbb{R}
}
\newcommand{\C}{
\mathbb{C}
}

\newcommand{\tp}{
\otimes
}

\newlength{\wurelwidth}
\newcommand{\urel}[2][=]{\mathrel{\mathop{#1}\limits_{\!\scalebox{0.5}{#2}\!}}}
\newcommand{\wurel}[2][=]{\mathrel{\mathop{#1}_{\!\scalebox{0.5}{\makebox[\the\wurelwidth]{#2}}\!}}}

\DeclareMathOperator{\sdd}{sdd}

\newcommand{\comps}{\pi_0}

\newcommand{\Graph}[2][1.0]{%
\vcenter{\hbox{\includegraphics[scale=#1]{Graphs/#2}}}%
}

\newcommand{\coeff}[2][]{c_{#2}^{#1}}
\newcommand{\rp}{\mu} 

\newcommand{\toyform}{\eta}

\newcommand{\scalelog}{\ell}
\newcommand{\x}{x}
\newcommand{\period}{\mathcal{P}}
\newcommand{\reg}{z} 
\newcommand{\toy}[1][]{\def\toyfirstArgument{#1}\toyinternal}
\newcommand{\toyinternal}[1][]{{^{}_{\reg}\phi^{#1}_{\toyfirstArgument}}}
\newcommand{\toyR}[1][]{{{}^{}_{\reg}\phi^{}_{R\ifthenelse{\equal{#1}{}}{}{,#1}}}}
\newcommand{\phiR}[1][]{\phi^{}_{\text{\tiny R}\ifthenelse{\equal{#1}{}}{}{,#1}}}
\newcommand{\MSR}[1][]{\def\MSRfirstArgument{#1}\MSRinternal}
\newcommand{\MSRinternal}[1][]{{^{}_{\reg}\phi^{#1}_{+\ifthenelse{\equal{\MSRfirstArgument}{}}{}{,\MSRfirstArgument}}}}
\newcommand{\MSRR}[1][]{\phi^{}_{+\ifthenelse{\equal{#1}{}}{}{,#1}}}
\newcommand{\MSZ}[1][]{\def\MSZfirstArgument{#1}\MSZinternal}
\newcommand{\MSZinternal}[1][]{{^{}_{\reg}\phi^{#1}_{-\ifthenelse{\equal{\MSZfirstArgument}{}}{}{,\MSZfirstArgument}}}}
\newcommand{\toylog}{\gamma} 

\newcommand{\intrules}[1][]{{^{#1}}\varphi}
\newcommand{\hide}[1]{}
\newcommand{\Rms}{R_{\text{\tiny MS}}} 
\newcommand{\bigo}[1]{\mathcal{O}\left(#1\right)}
\newcommand{\momsch}[1]{R_{#1}}

\newcommand{\coupling}{g}
\newcommand{\powdep}{\kappa}

\newcommand{\forests}{\mathcal{F}}
\newcommand{\f}{w}
\newcommand{\autoconc}{\triangleright} 
\newcommand{\unimor}[1]{{^{#1}}\!\rho} 
\newcommand{\trees}{\mathcal{T}}
\newcommand{\tree}[1]{%
\vcenter{\hbox{\includegraphics{Trees/#1}}}%
}

\newcommand{\convolution}{\star}
\newcommand{\counit}{\varepsilon}
\newcommand{\unit}{u}
\newcommand{\Prim}{
\mathrm{Prim}
}
\newcommand{\polyint}{\textstyle
\int_0
}

\newcommand{\cored}{\widetilde{\Delta}} 

\newcommand{\chars}[2]{{G}_{#2}^{#1}} 
\newcommand{\infchars}[2]{\mathfrak{g}_{#2}^{#1}} 
\newcommand{\auto}[1]{{^{#1}}\chi} 
\newcommand{\alg}[1][A]{\mathcal{#1}}
\newcommand{\decor}{\mathcal{D}}
\newcommand{\gradAut}{\theta}
\newcommand{\dH}{\delta} 
\newcommand{\HH}[1][]{\text{HH}^{#1}_{\counit}}
\newcommand{\HZ}[1][]{\text{HZ}^{#1}_{\counit}}
\newcommand{\HB}[1][]{\text{HB}^{#1}_{\counit}}

\newcommand{\dd}[1][]{\mathrm{d}^{#1}}
\newcommand{\restrict}[2]{%
{\left. #1 \right|}_{#2}%
}
\DeclareMathOperator{\Res}{Res}
\DeclareMathOperator{\Hom}{Hom}
\DeclareMathOperator{\lin}{lin}
\DeclareMathOperator{\Aut}{Aut}
\DeclareMathOperator{\End}{End}
\newcommand{\ev}{
\mathrm{ev}
}
\DeclareMathOperator{\im}{im}
\renewcommand{\1}{
\mathbbm{1}
}
\newcommand{\defas}{
\mathrel{\mathop:}=
}
\newcommand{\safed}{
=\mathrel{\mathop:}
}

\newcommand{\set}[1]{
\left\{ #1 \right\}
}
\newcommand{\setexp}[2]{
\left\{ #1\!:\ #2 \right\}
}
\newcommand{\abs}[1]{
\left\lvert #1 \right\rvert
}
\newcommand{\id}{
\mathrm{id}
}

\newcommand{\qft}{%
quantum field theory%
}
\newcommand{\qfts}{%
quantum field theories%
}
\newcommand{\momscheme}{kinematic subtraction scheme}

\newtheorem{satz}{Theorem}[section]
\newtheorem{definition}[satz]{Definition}
\newtheorem{lemma}[satz]{Lemma}
\newtheorem{korollar}[satz]{Corollary}
\newtheorem{proposition}[satz]{Proposition}
\newtheorem{beispiel}[satz]{Example}

\theoremstyle{definition}

\def\clap#1{\hbox to 0pt{\hss#1\hss}}

\def\mathclap{\mathpalette\mathclapinternal}

\def\mathclapinternal#1#2{%
\clap{$\mathsurround=0pt#1{#2}$}}

\hypersetup{%
	pdfauthor = {Erik Panzer}
}

\begin{document}

\title{Renormalization, Hopf algebras and Mellin transforms}


\author{Erik Panzer}
\address{%
Institutes of Physics and Mathematics,
Humboldt-Universit\"{a}t zu Berlin,
Unter den Linden 6, 10099 Berlin, Germany%
}
\email{\href{mailto:panzer@mathematik.hu-berlin.de}{\nolinkurl{panzer@mathematik.hu-berlin.de}}}

\subjclass[2000]{Primary }

\date{\today}

\begin{abstract}
	This article aims to give a short introduction into Hopf-algebraic aspects of renormalization, enjoying growing attention for more than a decade by now. As most available literature is concerned with the minimal subtraction scheme, we like to point out properties of the {\momscheme} which is also widely used in physics (under the names of MOM or BPHZ).

	In particular we relate renormalized Feynman rules $\phiR$ in this scheme to the universal property of the Hopf algebra $H_R$ of rooted trees, exhibiting a refined renormalization group equation which is equivalent to $\phiR: H_R \rightarrow \K[x]$ being a morphism of Hopf algebras to the polynomials in one indeterminate.
	
	Upon introduction of analytic regularization this results in efficient combinatorial recursions to calculate $\phiR$ in terms of the Mellin transform. We find that different Feynman rules are related by a distinguished class of Hopf algebra automorphisms of $H_R$ that arise naturally from Hochschild cohomology.

	Also we recall the known results for the minimal subtraction scheme and shed light on the interrelationship of both schemes.

	Finally we incorporate combinatorial Dyson-Schwinger equations to study the effects of renormalization on the physical meaningful correlation functions. This yields a precise formulation of the equivalence of the two different renormalization prescriptions mentioned before and allows for non-perturbative definitions of {\qfts} in special cases.
\end{abstract}

\maketitle

\section*{Motivation: The renormalization problem}
Suppose we want to assign a value to the logarithmically divergent integral $\phi_s(\tree{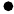}) \defas \int_0^\infty \frac{\dd x}{x + s}$, which we associate to the tree $\tree{+-}$. Observing the (absolutely) integrable difference
\begin{equation}
	\int_0^\infty \left[ \frac{\dd x}{x + s} - \frac{\dd x}{x + \rp} \right]
	=
	- \ln \frac{s}{\rp}
	\safed
	- \scalelog
	\label{eq:intro-()-def}
\end{equation}
allows for the definition of
$
	\phiR[s](\tree{+-})
	\defas
	\phi_s(\tree{+-}) - \phi_{\rp}(\tree{+-})
	= -\ln\frac{s}{\rp}
	= - \scalelog
$,
which we call the \emph{renormalized} value of the expression $\phi_s(\tree{+-})$. We need to choose the \emph{renormalization point} $\rp$ to fix the constant not determined by \eqref{eq:intro-()-def}. This natural \emph{renormalization scheme} given by subtraction at a \emph{reference scale} $s\mapsto\rp$ is commonly employed in {\qft} (where similar divergent expressions occur as we briefly describe in section \ref{sec:extensions}) and will be called \emph{\momscheme} in the sequel.

When we apply the same idea to multi-dimensional integrals, we have to take care of \emph{subdivergences} as occurring for example in the expression
\begin{equation}
	\phi_s\left( \tree{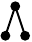} \right)
	\defas
	\int_0^\infty \frac{\dd x}{x+s}
	\left[ 
		\int_0^\infty \frac{\dd y}{x+y}
		\cdot
		\int_0^\infty \frac{\dd z}{x+z}
	\right]
	= \int_0^\infty \frac{\dd x}{x+s} \left[ \phi_x\left( \tree{+-} \right) \right]^2. 
	\label{eq:intro-(()())-def}
\end{equation}
A single subtraction at $s=\rp$ is insufficient as the sub integrals over $y$ and $z$ remain divergent. This problem is circumvented by applying renormalization to these first:
\begin{align}
	&\phiR[s]\left( \tree{++-+--} \right)
	\defas
	\int_0^\infty \left[\frac{\dd x}{x+s} - \frac{\dd x}{x+\rp} \right]
	\left\{
		\int_0^\infty \left[ \frac{\dd y}{x+y} - \frac{\dd y}{\rp + y} \right]
		\ \cdot\ 
		\int_0^\infty \left[ \frac{\dd z}{x+z} - \frac{\dd z}{\rp + z} \right]
	\right\} \nonumber\\
	&=\int_0^\infty \left[\frac{\dd x}{x+s} -\frac{\dd x}{x+\rp} \right] \left[ \phiR[x]\left( \tree{+-} \right) \right]^2
	= \int_0^\infty \left[\frac{\dd x}{x+s} -\frac{\dd x}{x+\rp} \right] \left(\ln \frac{x}{\rp} \right)^2
	= -\frac{\scalelog^3}{3} - \frac{\pi^2}{3}\scalelog.
	\label{eq:intro-(()())-R}
\end{align}
We want to summarize how this procedure is formulated in terms of Hopf algebras, study under which conditions it can be applied and reveal the main properties of the resulting maps $\phiR[s]$. In particular we will show that they are morphisms of Hopf algebras, taking values in the polynomials in $\scalelog$.

For a quick start, we prove this analytically in section \ref{sec:noregulator}, along the ideas \cite{Factorization} originating from {\qft}. Section \ref{sec:regulator} exploits an artificial regulator to rederive the same results in a setup more common to the original literature dealing with dimensional regularization and minimal subtraction. Along this way we take the time to recall the common algebraic techniques and contrast both methods.

After this construction of renormalized Feynman rules, we study their algebraic properties in section \ref{sec:RGE} focusing on the renormalization group. Together with the Mellin transform we can derive compact recursion relations, allowing for efficient combinatorial calculations.

At this point we turn towards the \emph{minimal subtraction scheme} in section \ref{sec:MS}. We summarize the known results and particularly relate the different realizations of the renormalization group equations in the two schemes, developing the duality between the concepts of \emph{finiteness} in the subtraction scheme and \emph{locality} in minimal subtraction.

Section \ref{sec:DSE} is devoted to Dyson-Schwinger equations, which link the combinatorics of the Hopf algebra to the physically meaningful correlation functions. In particular we observe how the change of renormalization scheme is equivalent to a redefinition of the coupling constant, proving the renormalization group equation in its physical form.

Finally we comment on the necessary modifications for generalizations of the model in different directions, like the presence of multiple parameters or higher degrees of divergence.

For reference and convenience of the reader, we collected the required features of the Hopf algebras $H_R$ of rooted trees and $\K[\x]$ of polynomials in the appendix. We also added a collection of well-known results on the \emph{Dynkin operator} $S\convolution Y$ which plays a prominent role in sections \ref{sec:regulator} and \ref{sec:MS} when we use a regulator. 

\section{Notations and preliminaries}
\label{sec:notation}

The essential structure behind perturbative renormalization is the Hopf algebra as discovered in \cite{Kreimer:HopfAlgebraQFT}.
As the literature grew comprehensive already, we content ourselves with fixing notation and recommend \cite{Manchon,Panzer:Master} for extended accounts of these concepts with a particular focus on their application to renormalization.

\subsection{Hopf algebras}
Throughout we consider associative, coassociative, commutative, unital and counital Bialgebras $(H,m,u,\Delta,\counit)$ given a connected ($H_0 = \K\cdot\1$) grading $H = \bigoplus_{n\geq 0} H_n$. 
For homogeneous $0\neq x\in H_n$, write $\abs{x}\defas n$ while the induced \emph{grading operator} $Y\in\End(H), x \mapsto Yx \defas \abs{x}\cdot x$ exponentiates to the one-parameter group $\K \ni t \mapsto \gradAut_t$ of Hopf algebra automorphisms
\begin{equation}
	\gradAut_t
	\defas \exp (tY)
	= \sum_{n\in\N_0} \frac{(tY)^n}{n!}
	,\quad
	\forall n\in\N_0:\quad
	H_n \ni x \mapsto \gradAut_t(x)
	= e^{t\abs{x}} x
	= e^{nt} x.
	\label{eq:grad-aut}
\end{equation}
Algebras $(\alg,m_{\alg},u_{\alg})$ are unital, associative and commutative, giving rise to the associative \emph{convolution product} on $\Hom(H,\alg)$ with unit given by $e\defas u_{\alg} \circ \counit$:
\begin{equation*}
	\Hom(H,\alg)
	\ni	\phi,\psi	\mapsto
	\phi \convolution \psi
	\defas
	m_{\alg} \circ (\phi\tp \psi) \circ \Delta
	\in \Hom(H,\alg).
\end{equation*}
As $H = \K\! \cdot\! \1 \oplus \ker \counit = \im \unit \oplus \ker \counit$ splits into the scalars and the \emph{augmentation ideal} $\ker\counit$, we obtain a projection $P \defas \id - \unit \circ \counit\!:\ H \twoheadrightarrow \ker \counit$ and use Sweedler's \cite{Sweedler} notation $\Delta (x) = \sum_x x_1 \tp x_2$ and $\cored(x)=\sum_x x'\tp x''$ to abbreviate the reduced coproduct $\cored \defas \Delta - \1 \tp \id - \id \tp \1$.
The connectedness implies:
\begin{enumerate}
	\item
		Under $\convolution$, the \emph{characters} (morphisms of unital algebras) form a group $\chars{H}{\alg} \defas \setexp{\phi\in\Hom(H,\alg)}{\phi\circ u=u_{\alg} \ \text{and}\ \phi\circ m = m_{\alg} \circ (\phi\tp \phi)}$.

	\item
		These biject along $\exp_{\convolution}\!:\ \infchars{H}{\alg}\rightarrow\chars{H}{\alg}$ with inverse $\log_{\convolution}\!:\ \chars{H}{\alg}\rightarrow\infchars{H}{\alg}$ to the \emph{infinitesimal characters} $\infchars{H}{\alg} \defas \setexp{\phi\in\Hom(H,\alg)}{\phi\circ m = \phi\tp e + e\tp \phi}$, using the pointwise finite series
	\begin{equation}\label{eq:log-exp}
		\exp_{\convolution} (\phi)
		\defas \sum_{n\in\N_0} \frac{\phi^{\convolution n}}{n!}
		\quad\text{and}\quad
		\log_{\convolution} (\phi)
		\defas \sum_{n\in\N} \frac{(-1)^{n+1}}{n} (\phi-e)^{\convolution n}.
	\end{equation}

	\item
		The unique inverse $S\defas\id^{\convolution-1} \in \chars{H}{H}$ is called \emph{antipode} and reveals $H$ as Hopf algebra. For all characters $\phi\in\chars{H}{\alg}$ we have $\phi^{\convolution -1} = \phi \circ S$.
\end{enumerate}
In general we assume the ground field $\K$ to be $\R$ or $\C$, though the reader will easily recognize that the majority of results allows for more generality (often characteristic zero suffices).
Note that by $\Hom(\cdot, \cdot)$ and $\End(\cdot)$ we always denote $\K$-linear maps and explicitly spell out if more structure should enjoy preservation by a morphism.
Finally, we write $\lin M$ for the linear span of $M$.

\subsection{Hochschild cohomology}

The Hochschild cochain complex \cite{CK:NC,BergbauerKreimer,Panzer:Master} we associate to $H$ contains the functionals $H' = \Hom(H,\K)$ as zero-cochains.
We will only consider one-cocycles $L \in \HZ[1](H) \subset \End(H)$ which are defined to solve $\Delta \circ L = (\id \tp L) \circ \Delta + L \tp \1$. The differential of this complex then becomes
\begin{equation}\label{eq:dH}
	\dH: H'\rightarrow \HZ[1](H),
	\alpha \mapsto
	\dH\alpha \defas (\id\tp\alpha) \circ \Delta - u \circ \alpha
	\in \HB[1](H)\defas \dH\left( H' \right)
\end{equation}
and defines the coboundaries $\HB[1](H)$ and thus the first cohomology group given by $\HH[1](H) \defas \HZ[1](H) / \HB[1](H)$. Note the elementary
\begin{lemma}\label{satz:cocycle-props}
	Cocycles $L\in \HZ[1](H)$ map into the augmentation ideal $\im L \subseteq \ker \counit$ and $L(\1) \in \Prim(H)\defas\ker\cored$ is primitive. The map $\HH[1](H) \rightarrow \Prim(H)$, $[L] = L + \HB[1](H) \mapsto L(\1)$ is well-defined since $\dH\alpha(\1)=0$ for all $\alpha\in H'$.
\end{lemma}

\section{Finiteness of renormalization by kinematic subtraction}
\label{sec:noregulator}

Originally, perturbative {\qft} assigns (divergent) expressions to combinatorial objects called \emph{Feynman graphs}, as we will comment on in section \ref{sec:extensions}. 
However the Hopf algebra $H_R$ of rooted trees summarized in appendix \ref{sec:H_R} suffices to encode the structure of subdivergences \cite{Kreimer:HopfAlgebraQFT,CK:NC,CK:RH1} such that we can focus on \emph{Feynman rules} of the form $\phi: H_R \rightarrow \alg$ as above.
The \emph{target algebra} $\alg$ has to sustain divergent expressions which only become finite after we accomplished the renormalization.
Therefore we consider $\alg$ as the integrands (differential forms) which for convenience we nevertheless write as integrals, keeping in mind that we do not evaluate them.

Guided by the examples \eqref{eq:intro-()-def} and \eqref{eq:intro-(()())-def} we make
\begin{definition}
	\label{def:unrenormiert}
	By virtue of \ref{satz:H_R-universal} let $\phi\in\chars{H_R}{\alg}$ be the character fixed through
	\begin{equation}
		\phi_s \left( B_+(\f) \right)
		\defas
		\int \frac{\dd \zeta}{s} f\left( \frac{\zeta}{s} \right) \phi_{\zeta} (\f)
		\quad\text{for any forest}\quad
		\f \in H_R,
		\label{eq:unrenormiert-def}
	\end{equation}
	where the function $f$ generalizes our choice of $f(\zeta) = \frac{1}{1+\zeta}$ in the introduction and will be dictated by the Feynman rules	in a physical application.
	Excluding infrared divergences by requiring $f$ to be bounded on $[0,\infty)$, we restrict our study to \emph{ultraviolet}\footnote{physically $\zeta$ corresponds to a momentum, so the limit $\zeta\rightarrow\infty$ means high energies} divergences at $\zeta \rightarrow \infty$.
\end{definition}
Observe how each node of a tree corresponds to an integration of $f$ times the function $\phi_{\zeta}(\f)$ given by its children, so \eqref{eq:unrenormiert-def} ensures that all information about subdivergences of these Feynman rules $\phi$ is encoded in the coproduct of $H_R$.
\begin{beispiel}
Looking at \eqref{eq:intro-(()())-def}, $\cored\left( \tree{++-+--} \right) = 2 \tree{+-} \tp \tree{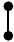} + \tree{+-}\tree{+-} \tp \tree{+-}$ informs us about:
\begin{enumerate}
	\item Two individual subdivergences $\int \frac{\dd y}{x+y} = \phi_x\left( \tree{+-} \right) = \int \frac{\dd z}{x + z}$ of the type $\phi\left(\tree{+-}\right)$ nested inside the outer integrals $\int_0^\infty \frac{\dd x}{x+s} \int_0^\infty \frac{\dd z}{x+z}$ and $\int_0^\infty \frac{\dd x}{x+s} \int_0^\infty \frac{\dd y}{x+y}$ respectively, each being of type $\phi\left(\tree{++--}\right)$.
	\item One subdivergence $\int \frac{\dd y}{x+y} \cdot \int \frac{\dd z}{x + z}$ of the kind $\phi\left(\tree{+-}\tree{+-}\right)$ (when $y$ and $z$ approach $\infty$ jointly) inside the outer integration $\phi_s\left( \tree{+-} \right) = \int_0^\infty \frac{\dd x}{x+s}$.
\end{enumerate}
\end{beispiel}

\subsection{Subtraction scheme}
Note that the integrands $\phi$ depend on a free \emph{external parameter} $s>0$ which models a physical mass or momentum.
Our goal is to replace $\phi$ by integrable integrands to achieve well-defined functions of $s$.
As exemplified in the introduction, this can be achieved by subtracting the integrand at $s \mapsto \rp$. 
Renormalizing the subdivergences first as in \eqref{eq:intro-(()())-R} motivates
\begin{definition}
	\label{def:renormiert}
	Fixing a renormalization point $\rp > 0$ we define the character $\phiR: H_R \rightarrow \alg$ (again as an instance of \ref{satz:H_R-universal}) by requiring
	\begin{equation}
		\phiR[s] (B_+(\f))
		\defas
		\int \dd \zeta \left[ \frac{f(\frac{\zeta}{s})}{s} - \frac{f(\frac{\zeta}{\rp})}{\rp} \right] \phiR[\zeta](\f)
		\quad\text{for any}\quad
		\f \in H_R.
		\label{eq:renormiert-def}
	\end{equation}
\end{definition}
To actually achieve finiteness this way we need to constrain the growth of $f(\zeta)$ at $\zeta\rightarrow \infty$ to be not worse than $\zeta^{-1}$, corresponding to a \emph{logarithmic divergence} in
\begin{satz}
	Let $f$ be a bounded, piecewise differentiable function on $[0,\infty)$ with $f(\zeta) \in \bigo{\zeta^{-1}}$, so in particular $\coeff{-1} \defas \lim_{\zeta \rightarrow \infty} \left[ \zeta f(\zeta) \right]$ exists. If furthermore
	\begin{equation}\label{eq:finiteness-conditions}
		f(\zeta) - \frac{\coeff{-1}}{\zeta},
		f(\zeta) + \zeta f'(\zeta)
		\in
 		\bigo{\zeta^{-1-\varepsilon}}
		\quad\text{for some}\quad
		\varepsilon>0,
	\end{equation}
	then for any forest $\f \in H_R$ the integral $\phiR[s](\f)$ is absolutely convergent and evaluates to a polynomial $\phiR[s](\f) \in \K[\scalelog]$ in $\scalelog \defas \ln\frac{s}{\rp}$.
	\label{satz:finiteness}
\end{satz}
	We remark that these requirements on $f$ could be relaxed%
\footnote{%
	for example, $f(\zeta) + \zeta f'(\zeta) \in \bigo{\zeta^{-1-\varepsilon}}$ already follows from $f(\zeta) - \frac{\coeff{-1}}{\zeta} \in \bigo{\zeta^{-1-\varepsilon}}$ using the theorem of L'H\^{o}pital
},
	but they are already fulfilled for physical cases of logarithmic divergent Feynman graphs%
\footnote{%
	these include the vertex graphs of QED and $\varphi^4$-theory in four dimensions of space-time
}.
	In particular \eqref{eq:finiteness-conditions} holds for all rational functions $f\in\bigo{\zeta^{-1}} \cap \K(\zeta)$ without poles in $[0,\infty)$.
\begin{proof}
	By definition \ref{def:renormiert} we have $\phiR(\1) = 1$ and $\phiR(\f \f') = \phiR(\f) \cdot \phiR(\f')$ and can therefore proceed inductively.
	We may thus assume the statement to be true for some element $\f\in H_R$ and only need to prove it for the tree $t=B_+(\f)$.
	But then the difference in brackets in \eqref{eq:renormiert-def} falls of like $\zeta^{-1-\varepsilon}$ from \eqref{eq:finiteness-conditions} while $\phiR[\zeta](\f)$ only grows like $\ln^N \zeta$ for the degree $N$ of $\phiR(\f)$. Hence \eqref{eq:renormiert-def} is absolutely convergent (the logarithmic singularities $\ln^N \zeta$ at $\zeta \rightarrow 0$ are integrable anyway) and thus $\phiR[s](B_+ (\f))$ finite.

	By \eqref{eq:finiteness-conditions} we can also interchange integration with the partial derivative $\partial_{\scalelog}$ in
	\begin{equation*}
		-\partial_{\scalelog}
		\phiR[s](t)
		= \int_0^\infty \!\!\!\!\!\dd\zeta \left[ 
			\frac{f(\frac{\zeta}{s})}{s}
			+
			\frac{\zeta}{s} \frac{f'(\frac{\zeta}{s})}{s}
		\right]
		\phiR[\zeta](\f)
		= \int_0^\infty \!\!\!\!\!\dd \zeta \left[ 
			f(\zeta) + \zeta f'(\zeta)
		\right]
		\phiR[\zeta s] (\f).
		\label{eq:diff-phi+-recursion}
	\end{equation*}
	Exploiting that $\phiR[\zeta s](\f)$ is polynomial in $\ln\frac{\zeta s}{\rp} = \ln\zeta + \scalelog$ we can evaluate
	\begin{equation}\label{eq:log-integration}
		\int_0^\infty \!\!\!\!\!\dd \zeta \left[ 
			f(\zeta) + \zeta f'(\zeta)
		\right]
		(\ln \zeta + \scalelog)^n
		=
		\sum_{i=0}^n \binom{n}{i} \scalelog^{n-i} c_{i-1} (-1)^i i!
		=
		\sum_{i=0}^\infty c_{i-1} \left(-\partial_{\scalelog}\right)^i \scalelog^n
	\end{equation}
	upon defining the constants (which are \emph{periods} \cite{periods} for algebraic functions $f$)
	\begin{equation}
		c_{n-1}
		\defas
		\int_0^\infty \dd \zeta \left[ 
			f(\zeta) + \zeta f'(\zeta)
		\right]
		\frac{(-\ln \zeta)^n}{n!}
		\qquad
		\text{for any}
		\qquad
		n \in \N_0.
		\label{eq:coeff-def}
	\end{equation}
	Thus linearity shows $\partial_{-\scalelog} \phiR[s] (t) \in \K[\scalelog]$ and we merely have to integrate once.
\end{proof}
Not only did we achieve our goal of renormalization, but we found an explicit recursion \eqref{eq:log-integration} determining $\phiR$ completely using the universal property \ref{satz:H_R-universal} in
\begin{korollar}
	The constants $\coeff{\cdot}$ of \eqref{eq:coeff-def} determine the renormalized Feynman rules $\phiR \in \chars{H_R}{\K[\scalelog]}$ completely through the universal property \ref{satz:H_R-universal} by
	\begin{equation}
		\phiR \circ B_+
		= P \circ F(-\partial_{\scalelog}) \circ \phiR,
		\quad\text{where}\quad
		F(-\partial_{\scalelog}) 
		\defas
		\sum_{n\geq -1} c_n (-\partial_{\scalelog})^n
		\in \End\left(\K[\scalelog]\right).
		\label{eq:phiR-universal}
	\end{equation}
	For convenience we write here $(-\partial_{\scalelog})^{-1} \defas - \polyint$ for the integral operator.  Recall that the projection $P: \K[\scalelog] \twoheadrightarrow \scalelog\,\K[\scalelog]$ annihilates any constants.
\end{korollar}
	Example \ref{ex:renormalized-universal} explicitly shows how this recursion works in detail.

	In section \ref{sec:RGE} we will see that \eqref{eq:phiR-universal} implies the \emph{renormalization group} upon realizing that $P \circ F (-\partial_{\scalelog}) \in \HZ[1](\K[\scalelog])$ is a Hochschild-1-cocycle.
But before let us review the

\subsection{Algebraic renormalization process}
\label{sec:renormalization}

Renormalization of a character $\phi\in\chars{H}{\alg}$ can be described as a \emph{Birkhoff decomposition} into the \emph{renormalized} $\phiR \defas \phi_+\in\chars{H}{\alg}$ and the \emph{counterterms} $\phi_-\in\chars{H}{\alg}$ subject to the conditions that
\begin{equation}
	\phi
	= \phi_-^{\convolution -1} \convolution \phi_+
	\quad\text{and}\quad
	\phi_{\pm} \left( \ker\counit \right) \subseteq \alg_{\pm}.
	\label{eq:birkhoff}
\end{equation}
It depends on a splitting \mbox{$\alg = \alg_+ \oplus \alg_-$} of the target algebra, determining the \emph{renormalization scheme} which we identify with the corresponding projection $R\!: \alg \twoheadrightarrow \alg_-$ along $\alg_+$.

\begin{satz}[\cite{CK:RH1,Manchon,Panzer:Master}]
	\label{satz:birkhoff}
A unique Birkhoff decomposition \eqref{eq:birkhoff} exists given that $R$ is a Rota-Baxter map, meaning
	\begin{equation}
		m \circ (R \tp R)
		=
		R \circ m\circ \left[ R \tp \id + \id \tp R - \id \tp \id \right].
		\label{eq:Rota-Baxter}
	\end{equation}
	On the augmentation ideal $\ker \counit$ it may be computed inductively by
	\begin{equation}
		\phi_-(x)
		=
		-R
		\circ
		\bar{\phi}(x)
		\qquad\text{and}\qquad
		\phi_+(x)
		=
		(\id-R)
		\circ
		\bar{\phi}(x),
		\label{eq:birkhoff-recursion}
	\end{equation}
	using the \emph{Bogoliubov character} $\bar{\phi}$ (also \emph{$\bar{R}$-operation}) which is defined as
	\begin{equation}
		\bar{\phi} (x) \defas
		\phi(x) + \sum_x \phi_- (x') \phi(x'') 
		= \phi_+(x) - \phi_-(x).
		\label{eq:bogoliubov}
	\end{equation}
\end{satz}
\begin{definition}\label{def:momentum-scheme}
	The \emph{\momscheme} $\momsch{\rp}$ by evaluation at $s \mapsto \rp$ is defined as
	\begin{equation}
		\End(\alg) \ni \momsch{\rp} \defas \ev_{\rp} = \left( \alg \ni f \mapsto {\left. f \right|}_{s=\rp} \right)
		\label{eq:momentum-scheme}
	\end{equation}
	and splits $\alg$ into $\im\momsch{\rp}=\alg_-$ ($s$-independent integrals) and $\ker\momsch{\rp} = \alg_+$, those integrals that vanish at $s=\rp$.
\end{definition}
As $\momsch{\rp}$ is a character of $\alg$, it not only fulfills \eqref{eq:Rota-Baxter} and we obtain a unique Birkhoff decomposition, but also the recursion \eqref{eq:birkhoff-recursion} simplifies a lot to just
\begin{equation}
	\label{eq:birkhoff-character}
	\phi_- 
	= \momsch{\rp} \circ \phi \circ S 
	= \phi_{\rp} \circ S
	= \phi_{\rp}^{\convolution -1}
	\qquad\text{and}\qquad
	\phi_+ 
	= \phi_{\rp}^{\convolution-1} \convolution \phi_{s}.
\end{equation}
\begin{beispiel}
	In accordance with \eqref{eq:intro-()-def} we find
	\begin{equation*}
		\phiR[s] \left( \tree{+-} \right)
		=
		\phi_{+,s} \left( \tree{+-} \right)
		=
		\phi_- \left( \tree{+-} \right) + \phi_s \left( \tree{+-} \right)
		= 
		\int_0^\infty \left[ - \restrict{\frac{\dd x}{x+s}}{s\mapsto\rp} + \frac{\dd x}{x+s} \right],
	\end{equation*}
	and 
	$%
		\bar{\phi}\left( \tree{++-+--} \right)
		\urel{\eqref{eq:birkhoff-recursion}}
		\phi_s\left( \tree{++-+--} \right) + 2\phi_-\left( \tree{+-} \right) \phi_s\left( \tree{++--} \right) + \phi_-\left( \tree{+-}\tree{+-} \right) \phi_s\left( \tree{+-} \right)
	$
	indeed agrees with \eqref{eq:intro-(()())-R} using $\phiR\left( \tree{++-+--} \right) = (\id-\momsch{\rp}) \bar{\phi}\left( \tree{++-+--} \right)$ after rearranging the terms\footnote{%
	Note that we need to track the correspondence of variables and nodes.}
	\begin{align*}
		\bar{\phi}\left( \tree{++-+--} \right)
		= \int_0^\infty \!\!\!\dd x
		  \int_0^\infty \!\!\!\dd y
		  \int_0^\infty \!\!\!\dd z
		& \left[ 
		\frac{1}{(s+x)(x+y)(x+z)} - \frac{1}{\rp+y}\frac{1}{(s+x)(x+z)}
			\right. \\
		&\quad \left.
		 - \frac{1}{\rp+z}\frac{1}{(s+x)(x+y)} + \frac{1}{\rp+y}\frac{1}{\rp+z}\frac{1}{s+x}
			\right].
	\end{align*}
\end{beispiel}
	We remark that the recursion \eqref{eq:birkhoff-recursion} makes explicit reference to the divergent counterterms $\phi_-$. In \eqref{eq:renormiert-def} we anticipated the much more practical formula resulting from the special structure \ref{satz:H_R-universal} of the Feynman rules $\phi$ of \eqref{eq:unrenormiert-def} in
\begin{satz}\label{satz:subdivergences}
	Let the character $\phi: H_R\rightarrow \alg$ be subject to $\phi\circ B_+ = L \circ \phi$ for some $L \in \End\left( \alg \right)$ and the renormalization scheme $R \in \End(\alg)$ such that it ensures
	\begin{equation}
		L \circ m_{\alg} \circ ( \phi_- \tp \id )
		= m_{\alg} \circ (\phi_- \tp L),
		\label{eq:counterterm-scalars}
	\end{equation}
	which means linearity of $L$ over the counterterms. Then we have
	\begin{equation}
		\bar{\phi} \circ B_+ = L \circ \phi_+
		\qquad\text{and therefore}\qquad
		\phi_+ \circ B_+ = (\id-R) \circ L \circ \phi_+
		\label{eq:rbar-cocycle}
	\end{equation}
\end{satz}
\begin{proof}
	This is a straightforward consequence of the cocycle property of $B_+$:
	\begin{align*}
		\bar{\phi} \circ B_+
		&= \left( \phi_- \convolution \phi - \phi_- \right) \circ B_+
		 = m_{\alg} \circ (\phi_- \tp \phi) \circ \left[ (\id \tp B_+) \circ \Delta + B_+ \tp \1 \right] 
		 		- \phi_- \circ B_+ \\
		&= \phi_- \convolution \left( \phi \circ B_+ \right) 
		 = \phi_- \convolution \left( L \circ \phi \right)
		\urel{\eqref{eq:counterterm-scalars}} L \circ \left( \phi_- \convolution \phi \right)
		 = L \circ \phi_+ \qedhere.
	\end{align*}
\end{proof}
As for $\momsch{\rp}$ the counterterms $\phi_-(x)\in\alg_-$ are independent of $s$, they separate from the integration in \eqref{eq:unrenormiert-def} and \eqref{eq:counterterm-scalars} is fulfilled indeed. This is a general feature of {\qfts}: The counterterms to not depend on any external variables%
\footnote{%
	Even if the divergence of a Feynman graph does depend on external momenta as happens for higher degrees of divergence, this dependence is only polynomial and stripped off by extracting the individual coefficients.
	In the Hopf algebra this can be encoded with \emph{external structures} which are given by distributions in \cite{CK:RH1}. So in any case, $\phi_-$ maps to constants.
}.

The significance of \eqref{eq:rbar-cocycle} lies in the expression of the renormalized $\phi_+(t)$ for a tree $t=B_+(\f)$ only in terms of the renormalized value $\phi_+(\f)$. This allows for inductive proofs like \ref{satz:finiteness} on properties of $\phiR=\phi_+$, without having to consider the unrenormalized Feynman rules or their counterterms (both of which are divergent) at all.

Summarizing, we proved in \ref{satz:finiteness} that for any forest $w\in H_R$, the expression $\phi_+(w) \in \alg_+$ is actually integrable and may be directly written as a convergent integral using \eqref{eq:renormiert-def}.

\section{Regularization and Mellin transforms}
\label{sec:regulator}

A technique often applied prior the renormalization is the introduction of a \emph{regulator} to assign finite values also to divergent expressions. Popular methods usually either alter the domain of integration:
\begin{enumerate}
	\item Confine integrations to the bounded interval $[0,\Lambda]$ for a \emph{cut-off} $\Lambda>0$. Then all integrals converge but acquire a dependence on $\Lambda$, which will in general diverge in the \emph{physical limit} $\Lambda\rightarrow\infty$ resembling the original situation. After renormalization however, this limit will be finite.

	\item Variations of mixed Hodge structures \cite{BlochKreimer:MixedHodge} also vary the chain of integration to avoid singularities.
\end{enumerate}
	or modify the integrand:
\begin{enumerate}\addtocounter{enumi}{2}
	\item Choose an \emph{analytic regulator} $0<\reg<1$ and replace each $\int_0^\infty \dd x$ with $\int_0^\infty x^{-\reg} \dd x$. This increases the decay of the integrand at $x\rightarrow \infty$ and we again get finite results which depend on $\reg$. As for the cut-off, these typically diverge in the physical limit $\reg \rightarrow 0$, unless we renormalize.

	\item \emph{Dimensional regularization} is similar in introducing a complex parameter $\reg \neq 0$ associated to a shifted dimension $D=4-2\reg$ of space-time. It is tailor made for Feynman integrals in {\qft} and we refer to \cite{Collins} for its definition and examples.
\end{enumerate}

We study the analytic regularization in detail, as it allows for the simplest algebraic description: Due to the regulator all integrals converge and give functions of both $s$ and $\reg$ that lie in the target algebra $\alg	= \K[\reg^{-1},\reg]] [s^{-\reg}]$ of Laurent series in $\reg$ as we shall see in proposition \ref{prop:regularized-mellin}.
\begin{definition}\label{def:regularized}
	The analytically regularized Feynman rules $\toy \in \chars{H_R}{\alg}$ are given through the universal property \ref{satz:H_R-universal} by requiring
	\begin{equation}\label{eq:regularized}
		\toy[s] \circ B_+ 
		= \int_0^{\infty} \frac{f(\frac{\zeta}{s})\zeta^{-\reg}}{s}\ \toy[\zeta] \ \dd\zeta 
		= \int_0^{\infty} f(\zeta)(s\zeta)^{-\reg}\ \toy[s\zeta] \ \dd\zeta.
	\end{equation}
\end{definition}
All these integrals can conveniently be evaluated in terms of the coefficients $\coeff{n}$ of the \emph{Mellin transform}\footnote{Conditions \eqref{eq:finiteness-conditions} suffice to prove that $F(\reg)$ is a Laurent series of this form.}
\begin{equation}\label{eq:mellin-trafo}
	F(\reg) \defas
	\int_0^{\infty} f(\zeta) \zeta^{-\reg} \ \dd\zeta
	= \sum_{n=-1}^{\infty} \coeff{n} {\reg}^n
	\in \reg^{-1} \K [[\reg]],
\end{equation}
which we already encountered in \eqref{eq:phiR-universal}: Indeed, a partial integration proves that
\begin{align*}
	&\coeff{n-1} n!
	\wurel{\eqref{eq:coeff-def}}
	\int_0^\infty \dd \zeta \left[ f(\zeta) + \zeta f'(\zeta) \right] (-\ln \zeta)^n
	=
	\restrict{\frac{\partial^n}{\partial\reg^n}}{\reg=0}
	\int_0^\infty \dd \zeta \left[ f(\zeta) + \zeta f'(\zeta) \right] \zeta^{-\reg} \\
	&= 
	\restrict{\frac{\partial^n}{\partial\reg^n}}{\reg=0}
	\left\{
		\left[ f(\zeta) \zeta^{1-\reg} \right]_{\zeta=0}^{\infty} 
		+
		\int_0^{\infty} \dd \zeta \left[ f(\zeta) - \left( 1 - \reg \right) f(\zeta) \right] \zeta^{-\reg}
	\right\}
	=
	\restrict{\frac{\partial^n}{\partial\reg^n}}{\reg=0}
	\left\{\reg F(\reg) \right\}.
\end{align*}
\begin{proposition}\label{prop:regularized-mellin}
	For any forest $\f \in \mathcal{F}$ we have (called \emph{BPHZ model} in \cite{BroadhurstKreimer:Auto})
	\begin{equation}\label{eq:regularized-mellin}
		\toy[s] (\f)
		= s^{-\reg\abs{\f}} \prod_{\mathclap{v \in V(\f)}} F \left( \reg \abs{\f_v} \right).
	\end{equation}
\end{proposition}
\begin{proof}
	As both sides of \eqref{eq:regularized-mellin} are clearly multiplicative, it is enough to inductively assume the claim for a forest $\f\in\forests$ and prove it for the tree $t=B_+(\f)$:
	\begin{align*}
		\toy[s] (t)
		&\wurel{\eqref{eq:regularized}} \int_0^{\infty} (s\zeta)^{-\reg} f(\zeta)\ \toy[s\zeta] (\f) \ \dd\zeta
		= \int_0^{\infty} (s\zeta)^{-\reg} f(\zeta) (s\zeta)^{-\reg\abs{\f}} \prod_{\mathclap{v \in V(\f)}} F \left( z \abs{\f_v} \right) \ \dd\zeta \\
		&= s^{-\reg \abs{B_+ (\f)}} \left[ \prod_{v \in V(\f)} F \left( \reg \abs{\f_v} \right) \right] F \left( \reg \abs{B_+ (\f)} \right)
		= s^{-\reg \abs{t}} \prod_{\mathclap{v \in V(t)}} F \left( \reg \abs{ t_v } \right) \qedhere.
	\end{align*}
\end{proof}
\begin{beispiel}
	Using \eqref{eq:regularized-mellin}, we can directly write down the Feynman rules like
\begin{equation*}
	\toy[s] \left( \tree{+-} \right)
	= s^{-\reg} F(\reg),
	\quad
	\toy[s] \left( \tree{++--} \right)
	= s^{-2\reg} F(\reg) F(2\reg)
	\quad\text{and}\quad
	\toy[s] \left( \tree{++-+--} \right)
	= s^{-3\reg} {\left[ F(\reg) \right]}^2 F(3\reg).
\end{equation*}
\end{beispiel}
Many examples (choices of $F$) are discussed in \cite{BroadhurstKreimer:Auto}, the particular case of the one-loop propagator graph $\gamma$ of Yukawa theory is in \cite{Kreimer:ExactDSE} and for scalar Yukawa theory in six dimensions one has
$
	F(\reg)
	=
	\frac{1}{\reg(1-\reg)(2-\reg)(3-\reg)}
$
as in \cite{Panzer:Master}. Already noted in \cite{Kreimer:ChenII}, the highest order pole of $\toy[s] (\f)$ is independent of $s$ and just the tree factorial
	\begin{equation}\label{eq:regularized-leading-pole}
		\toy[s] (\f)
		\in s^{-\reg\abs{\f}} \hspace{-1mm}
			\prod_{v \in V(\f)} \left\{
				\tfrac{\coeff{-1}}{\reg\abs{\f_v}} 
				+ \K[[\reg]]
			\right\} 
		\urel[\subset]{\eqref{eq:tree-factorial}}
		\frac{1}{\f!} {\left( \tfrac{\coeff{-1}}{\reg} \right)}^{\abs{\f}}
				+ \reg^{1-\abs{\f}} {\K}[[\reg]].
	\end{equation}

\subsection{Finiteness}
\label{sec:finiteness}

Using \eqref{eq:regularized-mellin} and \eqref{eq:birkhoff-character} we can quickly write down explicitly the values of the renormalized Feynman rules like in
\begin{beispiel}
	We find
	$
		\toyR[s] \left( \tree{+-} \right) 
		= \left( s^{-\reg} - \rp^{-\reg} \right) F(\reg)
	$
	and 
	$
		S\left( \tree{++--} \right) 
		= -\tree{++--} + \tree{+-}\tree{+-}
	$
	gives
\begin{equation}
	\toyR[s] \left( \tree{++--} \right)
		= \left( s^{-2\reg} - \rp^{-2\reg} \right) F(\reg)F(2\reg) - \left( s^{-\reg} - \rp^{-\reg} \right) \rp^{-\reg} F^2(\reg).
	\label{toyR:(())}
\end{equation}
\end{beispiel}
As the \emph{physical limit} $\reg\rightarrow 0$ reconstructs the original (unregularized) Feynman rules \eqref{eq:unrenormiert-def}, the finiteness of theorem \ref{satz:finiteness} is equivalent (by Lebesgue's theorem on dominated convergence) to the existence of the limit
\begin{equation}
	\phiR
	\defas
	\lim_{\reg \rightarrow 0} \toyR.
	\label{eq:physical-limit}
\end{equation}
\begin{korollar}\label{satz:holomorphie}
	The renormalized regularized Feynman rules are holomorphic, that is they map into
	$
		\im \left( \toyR[s] \right) \subset \K[[\reg]]
	$.
\end{korollar}
\begin{beispiel}\label{ex:phiR-reg-expansion}
Indeed we find $\toyR[s] \left( \tree{+-} \right) \in - \coeff{-1} \ln \tfrac{s}{\rp} + \reg\K[[\reg]]$. For \eqref{toyR:(())} check
\begin{align*}
	&\phiR \left( \tree{++--} \right)
	\urel{\eqref{eq:physical-limit}}
	\lim_{\reg \rightarrow 0} \left\{
		- \left[ -\reg \ln \tfrac{s}{\rp} + \tfrac{\reg^2}{2} \left( \ln^2 s + 2\ln s \ln\rp - 3\ln^2 \rp \right) \right] \cdot \left[ \tfrac{\coeff[2]{-1}}{\reg^2} + 2\tfrac{\coeff{-1} \coeff{0}}{\reg} \right] \right. \nonumber\\
	&	\quad + \left. \left[ -2\reg \ln \tfrac{s}{\rp} + 2\reg^2\left( \ln^2 s - \ln^2 \rp \right)\right]\cdot \left[ \tfrac{\coeff[2]{-1}}{2\reg^2} + \tfrac{3\coeff{0} \coeff{-1}}{2\reg} \right] \right\}
		= \frac{\coeff[2]{-1}}{2} \ln^2 \tfrac{s}{\rp} - \coeff{-1} \coeff{0} \ln \tfrac{s}{\rp},
\end{align*}
where all poles in $\reg$ perfectly cancel.
\end{beispiel}

Observe that we proved the by now purely combinatorial statement \ref{satz:holomorphie} of the cancellation of all pole terms in $\toyR$ analytically by estimates on the asymptotic growths in theorem \ref{satz:finiteness}. 
As we absorbed all analytic input of the integrands in $F(\reg)\in\reg^{-1}\K[[\reg]]$ in the series \eqref{eq:regularized-mellin} we can also give a completely combinatorial proof as we shall do in lemma \ref{satz:regularized-finiteness}.

Note that the analytic regularization yields a very simple dependence on the parameter $s$: Setting now $\alg\defas\C[\reg^{-1},\reg]]$ and 
$
	\toy 
	\defas
	{\toy[1]}
	= \restrict{\toy}{s=1}
	\in\chars{H_R}{\alg}
	$, \eqref{eq:regularized-mellin} fixes the scale dependence $\toy[s] = \toy \circ \gradAut_{-\reg\ln s}$ completely through the grading, see also \cite{Manchon,FardBondiaPatras:LieApproach}. Therefore we can write
\begin{equation}
	\label{eq:regularized-renormalized}
	\toyR[s]
	= \toy[\rp]^{\convolution -1} \convolution \toy[s]
	= \toy \circ \left[ (S\circ\gradAut_{-\reg\ln\rp}) \convolution \gradAut_{-\reg \ln s}  \right]
	= \toy \circ (S\convolution \gradAut_{-\reg\ln\frac{s}{\rp}}) \circ \gradAut_{-\reg\ln\rp}
\end{equation}
and characterize the finiteness of the physical limit \eqref{eq:physical-limit} in
\begin{proposition}\label{satz:finiteness-algebraic}
	For any character $\toy\in\chars{H_R}{\alg}$, the following are equivalent:
	\begin{enumerate}
		\item The physical limit $\phiR \defas \lim_{\reg\rightarrow 0} \toyR$ exists
		\item For any $\scalelog\in\K$, 
			$
				\toy^{\convolution -1} \convolution (\toy\circ\theta_{-\scalelog\reg}) 
				= \toy\circ (S\convolution \theta_{-\scalelog\reg})
			$
			maps into $\C[[\reg]]$.
		\item For every $n\in\N_0$, 
			$
				\toy^{\convolution -1} \convolution (\toy\circ Y^n)
				=
				\toy \circ (S\convolution Y^n)
			$ maps into $\reg^{-n}\C[[\reg]]$.
		\item $
				\toy^{\convolution -1} \convolution (\toy \circ Y)
				=
				\toy \circ (S\convolution Y)
			$ maps into $\frac{1}{\reg} \C[[\reg]]$, equivalently the limit 
			$
				\lim\limits_{\reg \rightarrow 0} \toy^{\convolution -1} \convolution (\toy \circ \reg Y)
			$ exists.
	\end{enumerate}
\end{proposition}
\begin{proof}
	From \eqref{eq:regularized-renormalized}, $(1) \Leftrightarrow (2)$ is just composition with the holomorphic $\gradAut_{-\reg\ln\rp}$ or $\gradAut_{\reg\ln\rp} = \gradAut_{-\reg\ln\rp}^{-1}$ while $(2) \Leftrightarrow (3)$ merely expands $\gradAut_{-\scalelog\reg} = \sum_{n\geq 0} \frac{(-\scalelog\reg Y)^n}{n!}$. It remains to prove $(4)\Rightarrow (3)$ inductively with
	\begin{equation*}
		\toy\circ \left( S\convolution Y^{n+1} \right)
		= \toy\circ (S\convolution Y^n) \circ Y 
			+ \left[\toy\circ(S\convolution Y)\right] \convolution \left[ \toy\circ (S\convolution Y^n) \right],
	\end{equation*}
	exploiting
	$
		(S\circ Y)\convolution\id
		= -S\convolution Y
	$
	in the formula ($\alpha$ arbitrary)
	\begin{equation*}
		S\convolution (\alpha\circ Y) -	(S\convolution\alpha)\circ Y
		= - (S\circ Y)\convolution\alpha
		= - \left[ (S\circ Y)\convolution\id \right]\convolution S \convolution \alpha \nonumber
	  =  S\convolution Y \convolution S \convolution \alpha. \qedhere
	\end{equation*}
\end{proof}

\begin{lemma}\label{satz:regularized-finiteness}
	Let $\toy \in \chars{H_R}{\alg}$ be the character defined by \eqref{eq:H_R-universal} with
	\begin{equation}\label{eq:regularized-general-mellin}
		\toy\circ B_+(\f) 
		= \toy(\f) \cdot F(\reg\abs{B_+(\f)})
		\qquad\text{for any fixed}\qquad
		F(\reg)\in \reg^{-1} \K[[\reg]].
	\end{equation}
	Then $\toy$ fulfills the conditions of proposition \ref{satz:finiteness-algebraic}. 
	In particular, $\im \left(\toyR \right)\subseteq \K[[\reg]]$ allows the finite physical limit 
	$
		\phiR 
		= \lim_{\reg\rightarrow 0} \toyR
		\subseteq \K[\scalelog, \coeff{\cdot}]
	$
	taking values in the polynomials in $\scalelog=\ln\frac{s}{\rp}$ and the coefficients $\coeff{n}$ of the series $F(\reg)$.
\end{lemma}
\begin{proof}
	We show $(2)$ of \ref{satz:finiteness-algebraic} inductively along the grading of $H_R$.
	So let it be true on $H_{R,m}$, then by the multiplicativity of $\toy \circ (S \convolution \gradAut_{-\reg\scalelog})$ it holds for all products in $H_{R,m+1}$ and we only need consider trees $t=B_+(\f)$ for some $\f \in H_{R,m}$. For any $k\in\N$ observe holomorphy of $\restrict{\partial_{-\scalelog}^k}{\scalelog=0} \toy \circ (S\convolution \gradAut_{-\reg\scalelog})$ through
	\begin{align}
		\toy \circ (S\convolution [\reg Y]^k) (t)
		&\urel{\eqref{eq:B_+-cocycle}}
		\toy \circ \left\{S \convolution ([\reg Y]^k \circ B_+) \right\}(\f)
		=
		\toy^{\convolution -1} \convolution (\toy \circ B_+ \circ [\reg (Y+\id)]^{k+1}) (\f) \nonumber\\
		&\urel{\eqref{eq:regularized-general-mellin}}
		\sum_{n\geq -1} \coeff{n} \cdot \toy \circ \left\{S \convolution \left[\reg ( Y + \id) \right]^{n+k} \right\} (\f) \nonumber\\
		&= \sum_{n \geq -1} \coeff{n} \sum_{j=0}^{n+k} \binom{n+k}{j} \reg^{n+k-j} \restrict{\partial_{-\scalelog}^j}{\scalelog=0} \toy \circ (S\convolution \gradAut_{-\reg\scalelog}) (\f)
		\in \K[[\reg]],
		\label{eq:toy-B_+-recursion}
	\end{align}
	while for $k=0$ we use $S\convolution [\reg Y]^0 = S \convolution \id = e$ and $e \circ B_+ = 0$.
\end{proof}

\subsection{Feynman rule recursion from Mellin transforms}
	In fact this serves an alternative prove of the recursion \eqref{eq:phiR-universal}, as in the physical limit $\reg \rightarrow 0$ only the contributions of $j=n+k$ in \eqref{eq:toy-B_+-recursion} survive:
\begin{equation}
	\phiR \circ B_+
	\wurel{\eqref{eq:regularized-renormalized}}
		\sum_{k\in\N} 
				\frac{(-\scalelog)^k}{k!} 
				\lim_{\reg \rightarrow 0} 
				\toy \circ (S\convolution [\reg Y]^k) \circ B_+
	\wurel{\eqref{eq:toy-B_+-recursion}}
			\sum_{\substack{
				k\in\N \\
				n \geq -1
			}}
			\coeff{n} \frac{(-\scalelog)^k}{k!}
			\left[
				\partial_{-\scalelog}^{n+k}
				\phiR \right]_{\scalelog=0}
	= P \circ F(-\partial_{\scalelog}) \circ \phiR.
	\label{eq:phiR-B_+-recursion-reg}
\end{equation}
Recall that $P=\id-\ev_{0}: \K[\scalelog] \twoheadrightarrow \ker \counit = \scalelog\,\K[\scalelog]$ projects out the constant terms and we defined $\partial_{\scalelog}^{-1} \defas \int_0$.
This delivers an efficient recursion to calculate $\phiR$ combinatorially in terms of the Mellin transform coefficients $\coeff{\cdot}$ without any need for series expansions in $\reg$ as in example \ref{ex:phiR-reg-expansion} or integrations like in \eqref{eq:renormiert-def}:
\begin{beispiel}\label{ex:renormalized-universal}
	Applying \eqref{eq:phiR-universal} we can reproduce example \ref{ex:phiR-reg-expansion} as
\begin{align*}
	\phiR \left( \tree{+-} \right)
	 &= \phiR \circ B_+ (\1)
	 = \left[
	 		P \circ F(-\partial_\scalelog)
			\right]
		\big( \phiR(1) \big)
	 = P \left( -\coeff{-1}\polyint 1 + \coeff{0} \right)
	 = - \coeff{-1}\, \scalelog \\
	\phiR \left( \tree{++--} \right)
	&= \phiR \circ B_+ \left( \tree{+-} \right)
	= \left[P \circ F(-\partial_{\scalelog}) \right]
		\big( \phiR \left( \tree{+-} \right) \big) \\
	&= P \left(
				\left[ -\coeff{-1}\polyint + \coeff{0} - \coeff{1}\partial_{\scalelog} \right] 
				\left( -\coeff{-1} \scalelog \right)
			\right)
 	= \coeff[2]{-1} \frac{\scalelog^2}{2} - \coeff{-1} \coeff{0}\, \scalelog, \\
\hide{
	 \phiR \left( \tree{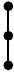} \right)
	&= \left[ -\coeff{-1}\polyint + \dH\toyform \right] \left( \coeff[2]{-1} \frac{\scalelog^2}{2} - \coeff{-1} \coeff{0}\, \scalelog \right)
	= - \coeff[3]{-1} \frac{\scalelog^3}{6} + \coeff[2]{-1} \coeff{0} \frac{\scalelog^2}{2}+ \coeff[2]{-1} \toyform(\scalelog)\,\scalelog \\
	 &\quad + \toyform(1) \left( \coeff[2]{-1} \frac{\scalelog^2}{2} - \coeff{-1} \coeff{0}\, \scalelog \right) 
	 = - \coeff[3]{-1} \frac{\scalelog^3}{6} + \coeff[2]{-1} \coeff{0}\, \scalelog^2 - \left( \coeff{-1} \coeff[2]{0} + \coeff[2]{-1} \coeff{1} \right) \scalelog,
	 \qquad\text{but also}\\
}%
	\phiR \left( \tree{++-+--} \right)
	&= \phiR \circ B_+ \left( \tree{+-}\tree{+-} \right)
	= \left[ P \circ F(-\partial_{\scalelog}) \right]
		\big( \phiR \left( \tree{+-}\tree{+-} \right) \big)
	= P \circ \left[ -\coeff{-1}\polyint + \coeff{0} - \coeff{1} \partial_{\scalelog} \right] \left\{ {\left( -\coeff{-1}\,\scalelog \right) }^2 \right\} \\
	&= -\coeff[3]{-1} \frac{\scalelog^3}{3} + \coeff[2]{-1} \coeff{0}\,\scalelog^2 - 2\coeff[2]{-1}\coeff{1}\,\scalelog.
\end{align*}
Here we can substitute $\coeff{-1} = 1$, $\coeff{0}=0$ and $\coeff{1}=\zeta(2)=\frac{\pi^2}{6}$ to finally verify \eqref{eq:intro-(()())-R} from the introduction, where the choice $f(\zeta) = \frac{1}{1+\zeta}$ results in the beta function
\begin{equation*}
	F(\reg)
	\urel{\eqref{eq:mellin-trafo}} 
	B(\reg,1-\reg) 
	= \Gamma(\reg) \Gamma(1-\reg) 
	= \frac{\pi}{\sin(\pi\reg)} 
	\in \reg^{-1}+\frac{\pi^2}{6} \reg + \frac{7\pi^4}{360} \reg^3 + \bigo{\reg^5}.
\end{equation*}
\end{beispiel}

\begin{korollar}\label{satz:phiR-leading-log}
	As in $F(-\partial_{\scalelog})$ only $-\coeff{-1}\polyint$ increases the degree in $\scalelog$, the highest order contribution (called \emph{leading $\log$}) of $\phiR$ is the tree factorial we already saw in \eqref{eq:regularized-leading-pole}: For any forest $\f \in \forests$,
	\begin{equation}
		\phiR (\f) \in
		\unimor{\left[-\coeff{-1}\polyint\right]} (\f) + \bigo{\x^{\abs{\f}-1}}
		\urel{\ref{ex:tree-factorial}}
		\frac{ {\left( - \coeff{-1} \x \right)}^{\abs{\f}}}{\f!} + \K[\x]_{<\abs{\f}}.
	\end{equation}
\end{korollar}

\section{Hopf algebra morphisms and the renormalization group}
\label{sec:RGE}

From now on we identify $\phiR\!: H_R\rightarrow\K[\x]$ with the polynomials that evaluate to the renormalized Feynman rules $\phiR[s] = \ev_{\scalelog} \circ \phiR$ at $\x \mapsto \scalelog=\ln\frac{s}{\rp}$. In \eqref{eq:phiR-universal} and \eqref{eq:phiR-B_+-recursion-reg} we independently proved
\begin{korollar}\label{satz:toyphy-hopfmor}
	As $P \circ F(-\partial_{\x}) \in \HZ[1](\K[\x])$ in \eqref{eq:phiR-universal} is a Hochschild-1-cocycle by \eqref{eq-polys-cycles} and \eqref{satz:poly-coboundaries}, theorem \ref{satz:H_R-universal} implies that $\phiR\!: H_R \rightarrow \K[\x]$ is a morphism of Hopf algebras.
\end{korollar}
Therefore $\Delta \circ \phiR = (\phiR \tp \phiR) \circ \Delta$ and the induced map $\chars{{\K}[\x]}{\K}\rightarrow \chars{H_R}{\K}$ given by $\ev_{\scalelog} \mapsto \restrict{\phiR}{\scalelog} \defas \ev_{\scalelog} \circ \phiR$ becomes a morphism of groups, implying
\begin{korollar}\label{satz:rge}
	Using \eqref{eq:poly-characters} we obtain the \emph{renormalization group equation} (called \emph{Chen's lemma} in \cite{Kreimer:ChenII})
	\begin{equation}\label{eq:rge}
		\restrict{\phiR}{\scalelog}	\convolution	\restrict{\phiR}{\scalelog'}
		= \restrict{\phiR}{\scalelog+\scalelog'},
	\quad\text{for any}\quad
	\scalelog,\scalelog' \in \K.
	\end{equation}
\end{korollar}
Before we obtain the generator of this one-parameter group in \ref{def:toylog}, note how this result imposes non-trivial relations between individual trees like
\begin{align*}
	\phiR[\scalelog] \convolution \phiR[\scalelog'] \left( \tree{++--} \right)
	&\wurel{\eqref{eq:rge}}
		\phiR[\scalelog] \left( \tree{++--} \right)
		+ \phiR[\scalelog] \left( \tree{+-} \right) \phiR[\scalelog'] \left( \tree{+-} \right)
		+ \phiR[\scalelog'] \left( \tree{++--} \right) \\
	&\wurel{\eqref{ex:renormalized-universal}}
		\coeff[2]{-1} \frac{\scalelog^2+\scalelog'^2}{2}
		- \coeff{-1}\coeff{0} \left(\scalelog + \scalelog' \right)
		+ \coeff[2]{-1}\scalelog \scalelog'
	\urel{\eqref{ex:renormalized-universal}}
		\phiR[\scalelog+\scalelog'] \left( \tree{++--} \right).
\end{align*}
\begin{proposition}
	Let $\phi\!: H \rightarrow \K[\x]$ be a morphism of bialgebras\footnote{This already implies $\phi$ to be a morphism of Hopf algebras since $H$ is connected.}, then $\log_\convolution \phi$ is precisely the monomial linear in $\x$:
	\begin{equation}\label{eq:log-phi}
		\log_{\convolution} \phi
		= \x \cdot \left[ \partial_0 \circ \phi \right].
	\end{equation}
\end{proposition}
\begin{proof}
	Letting $\phi\!: C\rightarrow H$ and $\psi\!: H\rightarrow \alg$ denote morphisms of coalgebras and algebras, exploiting
$
	\left( \psi \circ \phi - \unit_{\alg} \circ \counit_C \right)^{\convolution n}
		= \psi \circ \left( \phi - \unit_H\circ\counit_H \right)^{\convolution n}
		= \left( \psi - \unit_{\alg}\circ\counit_H \right)^{\convolution n} \circ \phi 
$
in \eqref{eq:log-exp} proves 
$
	(\log_{\convolution} \psi) \circ \phi 
	= \log_{\convolution} (\psi \circ \phi) 
	= \psi \circ \log_{\convolution} \phi
$. Now set $\psi=\ev_a$ and use lemma \ref{satz:poly-log}.
\end{proof}
\begin{definition}\label{def:toylog}
	The \emph{anomalous dimension} $\toylog$ of $\phiR$ is the infinitesimal character
	\begin{equation}\label{eq:toylog}
		H_R' \supset
		\infchars{H_R}{\K} \ni \toylog
		\defas -\partial_0 \: \circ\: \phiR
		\urel{\eqref{eq:log-phi}} 
		-\frac{1}{\x} \log_{\convolution} \phiR.
	\end{equation}
	It completely determines all higher powers of $\x$ by means of
	\begin{equation}\label{eq:toyexp}
		\phiR
			= \exp_{\convolution}(-\x\cdot\toylog)
			= \sum_{n\in\N_0} \frac{\toylog^{\convolution n}}{n!} (-\x)^n.
	\end{equation}
\end{definition}
\begin{beispiel}
	Reading off
$
	\toylog \left( \tree{+-} \right)
	= \coeff{-1}
$,
$
	\toylog \left( \tree{++--} \right)
	= \coeff{-1}\coeff{0}
$
and
$
	\toylog \left( \tree{++-+--} \right)
	= 2\coeff[2]{-1}\coeff{1}
$
from the example \ref{ex:renormalized-universal} above, \eqref{eq:toyexp} correctly determines the higher powers of $\x$ through
\begin{align*}
	\phiR \left( \tree{++--} \right)
	&
	 \urel{\eqref{eq:log-exp}} \left[ e - \x\toylog + \x^2\frac{\toylog\convolution\toylog}{2} \right] \left( \tree{++--} \right)
	 = 0 - \x\toylog\left( \tree{++--} \right) + \x^2 \frac{\toylog^2\left( \tree{+-} \right)}{2} 
	 = - \coeff{-1}\coeff{0}\,\x + \coeff[2]{-1} \frac{\x^2}{2}, \\
	\phiR \left( \tree{++-+--} \right)
	&= 0 - \x \toylog \left( \tree{++-+--} \right) 
		+ \x^2\frac{\toylog \tp \toylog}{2} \left( 2 \tree{+-} \tp \tree{++--} + \tree{+-}\tree{+-} \tp \tree{+-} \right)
		- \x^3\frac{\toylog \tp \toylog \tp \toylog}{6} \left( 2\tree{+-} \tp \tree{+-} \tp \tree{+-} \right) \\
	&= -\toylog^3\left(\tree{+-}\right) \frac{\x^3}{3} + \x^2 \toylog\left(\tree{+-}\right) \toylog\left(\tree{++--}\right)  - 2\coeff[2]{-1}\coeff{1}\,\x
	 = -\coeff[3]{-1} \frac{\x^3}{3} + \coeff[2]{-1} \coeff{0}\,\x^2 - 2\coeff[2]{-1}\coeff{1}\,\x.
\end{align*}
\end{beispiel}
Note how the fragment $\tree{+-}\tree{+-} \tp \tree{+-}$ of $\Delta \left( \tree{++-+--} \right)$ does not contribute to the quadratic terms $\frac{\x^2}{2} \toylog \convolution \toylog$, as $\toylog$ vanishes on products. We will exploit this in \eqref{eq:perturbation-convolution-inf} of section \ref{sec:propagator-coupling} on Dyson-Schwinger equations.
\begin{beispiel}\label{ex:intrules-exp}
	In the leading-$\log$ case \eqref{eq:int-rules} we read off $\partial_0 \circ \intrules = Z_{\tree{+-}} \in \infchars{H_R}{\K}$ where $Z_{\tree{+-}} (\f) \defas \delta_{\f,\tree{+-}}$. Comparing $\intrules = \exp_{\convolution}(xZ_{\tree{+-}})$ with \eqref{eq:int-rules} shows $\abs{\f}! = \f! \cdot Z_{\tree{+-}}^{\convolution \abs{\f}}(\f)$, implying the following combinatorial relation among tree factorials noted in \cite{Kreimer:ChenII}:
	\begin{equation*}
		\frac{\abs{\f}}{\f!}
		=	\frac{1}{\left( \abs{\f} -1 \right)!} \sum_{\f} Z_{\tree{+-}}(\f_1)Z_{\tree{+-}}^{\convolution \abs{\f}-1}(\f_2)
		= \sum_{\f:\ \f_1=\tree{+-}}
			\frac{1}{\abs{\f_2}!}Z_{\tree{+-}}^{\convolution \abs{\f_2}} (\f_2)
		= \sum_{\f:\ \f_1=\tree{+-}}
			\frac{1}{\f_2!}.
	\end{equation*}
	This Sweedler sum picks only the parts with $\f_1 = \tree{+-}$ of the coproduct, which means that we sum over all leaves (nodes without children) $\f_1$ of $\f$ and $\f_2$ is obtained by cutting off this leaf.
\end{beispiel}

\subsection{The regularized viewpoint}
	We can obtain these results also by exploiting the regulator as in \cite{CK:RH2}:
\begin{lemma}
	For $\toy \in \chars{H}{\alg}$ subject to \ref{satz:finiteness-algebraic}, the \emph{anomalous dimension}
	\begin{equation}
		\toylog
		\defas -\restrict{\partial_\scalelog}{0} \phiR
		= \lim_{\reg \rightarrow 0}
				\reg \cdot \toy \circ (S\convolution Y)
		= \Res \toy \circ (S\convolution Y)
		\in \infchars{H}{\K}
	\end{equation}
	is the residue (coefficient of $\frac{1}{\reg}$) of $\toy \circ (S \convolution Y) \in \reg^{-1} \K[[\reg]]$ and fulfills
	\begin{equation}
		-\frac{\partial}{\partial \scalelog} \phiR
		= \toylog \convolution \phiR
		,\quad\text{therefore}\quad
		\phiR
		= \exp_{\convolution} (-\scalelog \toylog) \convolution \restrict{\phiR}{\scalelog=0}
		= \exp_{\convolution} (-\scalelog \toylog).
	\end{equation}
\end{lemma}
\begin{proof}
	Proposition \ref{satz:S*Y-infchar} renders $\toylog\in\infchars{H}{\K}$ immediate.
	We employ the coderivation property $\Delta \circ Y = (\id \tp Y + Y \tp \id) \circ \Delta$ and $S \circ Y = - S \convolution Y \convolution S$ in
	\begin{align*}
		&-\frac{\partial}{\partial \scalelog} \lim_{\reg \rightarrow 0} \toy \circ (S\convolution\gradAut_{-\reg\scalelog})
		= \lim_{\reg \rightarrow 0} \reg \cdot \toy \circ \left( S \convolution [\gradAut_{-\reg\scalelog}\circ Y] \right) \\
		&= \lim_{\reg \rightarrow 0} \Big\{
				\reg \cdot \toy \circ \left( S \convolution \gradAut_{-\reg\scalelog} \right) \circ Y 
				\Big\}
				+\lim_{\reg\rightarrow 0} \Big\{
					\reg \cdot \toy \circ (S \convolution Y) \convolution \toy \circ \left( S \convolution \gradAut_{-\reg\scalelog} \right)
			\Big\}.
	\end{align*}
	The first term vanishes by the existence of $\lim_{\reg \rightarrow 0} \toy \circ (S \convolution \gradAut_{-\reg\scalelog})$, while the second factorizes as desired.
	It remains to observe $\restrict{\phiR}{\scalelog=0} = \toy \circ (S\convolution \id) = \toy \circ e = e$.
\end{proof}
	Clearly we can easily rewrite this in the form of \eqref{eq:rge} since
	\begin{equation}
		\phiR[\scalelog] \convolution \phiR[\scalelog']
		=	\exp_{\convolution}(-\scalelog\toylog) \convolution \exp_{\convolution}(-\scalelog' \toylog)
		= \exp_{\convolution}(-(\scalelog+\scalelog')\toylog)
		= \phiR[\scalelog+\scalelog'].
		\label{eq:rge-from-exp}
	\end{equation}
Note how this reasoning fails if $\phiR$ disrespects the coproduct: Then
\begin{equation}
	\log_{\convolution} \phiR
	=
	\sum_{n\in\N} \frac{\toylog_n}{n!} \x^n
\end{equation}
would contain higher powers in $\x$ and a family $\toylog_n \in H'$ of functionals. As these do not necessarily commute under $\convolution$, also $\ev_{\scalelog} \circ \log_{\convolution} (\phiR)$ and $\ev_{\scalelog'} \circ \log_{\convolution} (\phiR)$ do not commute such that \eqref{eq:rge-from-exp} is not applicable.

While the renormalization group allows us to reduce all computations to the linear terms $\toylog$, in our setup \eqref{eq:regularized} we can give a simple recursion for $\toylog$ itself in term of the Mellin transform coefficients in
\begin{korollar}\label{satz:toylog-mellin-recursion}
	From
	$
		\toylog\circ B_+ 
		\urel{\eqref{eq:phiR-universal}}
		-\partial_0 \circ F(-\partial_{\scalelog}) \circ \phiR 
		\urel{\eqref{eq:toyexp}}
		\ev_0 \circ [\reg F(\reg)]_{-\partial_{\x}} \circ\exp_{\convolution}(-\x\toylog)
	$
	we obtain the inductive formula
	$
		\toylog \circ B_+
		= \sum_{n\in\N_0} \coeff{n-1} \toylog^{\convolution n}
		= \left[ \reg F(\reg) \right]_{\reg \mapsto \toylog}.
	$
\end{korollar}
As $\toylog\in\infchars{H}{\K}$ vanishes on products, evaluating it on trees is all we need such that \ref{satz:toylog-mellin-recursion} is sufficient to determine $\toylog$.
\begin{beispiel}\label{ex:toylog-mellin-recursion}
	Starting with
	$
		\toylog\left(\tree{+-} \right)
		= \coeff{-1} \counit(\1) = \coeff{-1}
	$
	we can recursively calculate 
	\begin{align*}
		\toylog\left( \tree{++--} \right)
		&= \coeff{-1}\counit\left( \tree{+-} \right) + \coeff{0} \toylog\left( \tree{+-} \right)
		= \coeff{-1}\coeff{0}, \\
		\toylog\left( \tree{+++---} \right)
		&= \coeff{-1}\counit\left( \tree{++--} \right) + \coeff{0}\toylog\left( \tree{++--} \right) + \coeff{1}\toylog\convolution\toylog\left( \tree{++--} \right)
		= \coeff{-1}\coeff[2]{0} + \coeff{1} \left[ \toylog\left( \tree{+-} \right) \right]^2
		= \coeff{-1}\coeff[2]{0} + \coeff[2]{-1}\coeff{1}, \\
		\toylog\left( \tree{++-+--} \right)
		&= \coeff{-1}\counit\left( \tree{+-} \tree{+-} \right)
			+\coeff{0}\toylog\left( \tree{+-} \tree{+-} \right)
			+\coeff{1}\toylog\convolution\toylog\left( \tree{+-} \tree{+-} \right)
		= 2\coeff{1} \left[ \toylog\left( \tree{+-} \right) \right]^2
		= 2\coeff[2]{-1}\coeff{1}
		\quad\text{and so on.}
	\end{align*}
\end{beispiel}

\section{Locality, finiteness and minimal subtraction}
\label{sec:MS}

In the presence of a regulator $\reg$, considering a general character $\toy \in \chars{H}{\alg}$ for $\alg = \K[\reg^{-1},\reg]]$ and scale dependence fixed by $\toy[s] = \toy \circ \gradAut_{-\reg\scalelog}$ as in section \ref{sec:regulator} naturally leads to a different idea of renormalization in
\begin{definition}\label{def:MS}
	The \emph{minimal subtraction scheme} $\Rms$ splits $\alg = \alg_- \oplus \alg_+$ by projection on the poles $\alg_- \defas \reg^{-1}\K[\reg^{-1}]$ along the holomorphic $\alg_+ \defas \K[[\reg]]$.

	One easily checks \eqref{eq:Rota-Baxter} and therefore obtains a unique Birkhoff decomposition $\MSR = \MSZ \convolution (\toy \circ \gradAut_{-\reg\scalelog})$, with physical limit $\MSRR \defas \lim_{\reg\rightarrow 0} \MSR$.
\end{definition}
This setup is the starting point in \cite{CK:RH2} and subject to many articles like \cite{Manchon,FardBondiaPatras:LieApproach}.
Note that this scheme does not specify a subtraction point $\rp$, but we included $\rp$ as an arbitrary scale inside $\scalelog = \ln \frac{s}{\rp}$ when we agreed on $\toy[s] = \toy \circ \gradAut_{-\reg\scalelog}$. 
Physically this is necessary to obtain the dimensionless argument $\frac{s}{\rp}$ in the logarithm instead of expressions like $\ln s$ alone, as $s$ is a quantity carrying a unit (typically momentum or energy). In any case, replacing $\ln s$ by $\scalelog$ is nothing but a rescaling of $s$.

This renormalization scheme $\Rms$ and the resulting Birkhoff decomposition differ from the {\momscheme} $\momsch{\rp}$, compare example \ref{ex:renormalized-universal} with
\begin{beispiel}\label{ex:MS-R}
	For $\toy$ arising from \eqref{eq:regularized-general-mellin}, minimal subtraction yields
	\begin{align}
		\MSR \left( \tree{+-} \right)
		&= (\id - \Rms) \toy[s] \left( \tree{+-} \right) 
		= (\id - \Rms) e^{-\reg\scalelog} F(\reg)
		= e^{-\reg\scalelog} F(\reg) - \frac{\coeff{-1}}{\reg}
		\nonumber\\
		\MSRR \left( \tree{+-} \right)
		&= \lim_{\reg \rightarrow 0} \MSR\left( \tree{+-} \right)
		= \coeff{0} - \coeff{-1} \scalelog 
		\label{eq:MS-()-R}\\
		\MSR\left( \tree{++--} \right)
		&= (\id - \Rms) \left[ \toy[s]\left( \tree{++--} \right) + \MSZ\left( \tree{+-} \right) \toy[s]\left( \tree{+-} \right) \right]
		\nonumber\\
		&= (\id - \Rms) \left[ e^{-2\reg\scalelog} F(\reg) F(2\reg) - \frac{\coeff{-1}}{\reg} e^{-\reg\scalelog} F(\reg) \right] 
		\nonumber\\
		&= e^{-2\reg\scalelog} F(\reg) F(2\reg) - \frac{\coeff{-1}}{\reg} e^{-\reg\scalelog} F(\reg) + \frac{\coeff[2]{-1}}{2\reg^2} - \frac{\coeff{-1}\coeff{0}}{2\reg} 
		\nonumber\\
		\MSRR\left( \tree{++--} \right)
		&= \frac{\coeff[2]{-1}}{2} \scalelog^2 - 2\coeff{-1}\coeff{0}\scalelog + \coeff[2]{0} + \frac{3}{2}\coeff{-1}\coeff{1}.
		\label{eq:MS-(())-RR}
	\end{align}
\end{beispiel}
Note that by this choice of $\alg_+ = \K[[\reg]]$, the finiteness of $\MSRR \defas \lim_{\reg \rightarrow 0} \MSR$ is automatic such that we can finitely renormalize any $\toy \in \chars{H}{\alg}$ using $\Rms$.
This seems preferable considering that the {\momscheme} only yields finite results under the conditions of proposition \ref{satz:finiteness-algebraic}. 
However, the physics of local field theory requires \emph{local} counterterms that are constants, independent of the external parameters (in our setup $s$)%
\footnote{%
See also section \ref{sec:extensions}: Counterterms may depend polynomially on parameters, but not logarithmically as might in general happen for $\im \left( \MSZ[s] \right) \subseteq \reg^{-1}\K[\reg^{-1}, \scalelog]$. In practice one actually defines \emph{form-factors} as the coefficients of these polynomials and therefore become indeed completely independent of kinematics.%
}.
\begin{definition}[\cite{CK:RH2}]
	A Feynman rule $\toy \in \chars{H_R}{\alg}$ is called \emph{local} iff the minimal subtraction counterterm $\MSZ[s] = (\toy \circ \theta_{-\reg \scalelog})_-$ is independent of $\scalelog\in\K$.
\end{definition}
	By definition \ref{def:momentum-scheme}, counterterms in the {\momscheme} are $s$-independent a priori. 
	For the minimal subtraction scheme $\Rms$, locality is a true condition and the study and characterization of local Feynman rules in this setting is a main theme of \cite{CK:RH2,Manchon}. 
	It is therefore illuminating to find
\begin{proposition}\label{satz:finiteness=locality}
	Locality of $\toy\in\chars{H_R}{\alg}$ in the $\Rms$ scheme is equivalent to the finiteness conditions of proposition \ref{satz:finiteness-algebraic} in the \momscheme.
\end{proposition}
\begin{proof}
	Given condition \ref{satz:finiteness-algebraic} (2), 
	$
				\MSR \convolution
				[ \toy^{\convolution -1} \convolution \left( \toy\circ\gradAut_{-\reg\scalelog} \right) ]
	$
	maps to $\K[[\reg]]$ so
	\begin{equation*}
		\toy \circ \gradAut_{-\reg\scalelog}
		= \MSZ^{\convolution -1} \convolution \big\{
				\MSR \convolution
				[ \toy^{\convolution -1} \convolution \left( \toy\circ\gradAut_{-\reg\scalelog} \right) ]
				\big\}
			\end{equation*}
	is a Birkhoff decomposition and its uniqueness implies the locality $\left( \toy \circ \gradAut_{-\reg\scalelog} \right)_- = \MSZ$. Conversely, for local $\toy$ we have $\MSZ = \left(\toy \circ \gradAut_{-\reg\scalelog} \right)_-$ wherefore
	\begin{equation*}
		\MSR^{\convolution -1} \convolution \left(\toy \circ \gradAut_{-\reg\scalelog}\right)_+
		= \toy^{\convolution -1} \convolution \MSZ^{\convolution -1}
			\convolution
			\left(\toy \circ \gradAut_{-\reg\scalelog} \right)_-
			\convolution
				(\toy \circ \gradAut_{-\reg\scalelog})
		=	\toy^{\convolution - 1} \convolution (\toy \circ \gradAut_{-\reg\scalelog})
	\end{equation*}
	shows \ref{satz:finiteness-algebraic} (2) as the left hand side maps to $\alg_+ = \K[[\reg]]$.
\end{proof}
We have seen how algebraically the problems of finiteness in the {\momscheme} and locality in minimal subtraction coincide.
Both finite renormalization and local counterterms are simultaneously only achieved under the conditions of proposition \ref{satz:finiteness-algebraic}, no matter which of the schemes $\set{\momsch{\rp}, \Rms}$ is chosen.

\subsection{Renormalization group}
We again identify $\MSRR$ with the polynomials in $\K[\x]$ such that $\MSRR[s] = \ev_{\scalelog} \circ \MSRR$, but in contrast to $\phiR$ from the \momscheme, these feature constant terms $\MSRR[\rp] = \restrict{\MSRR}{\x=0}=\ev_0 \circ \MSRR = \counit \circ \MSRR$ as we observed in the examples \ref{ex:MS-R}.
Therefore the renormalization group equation \eqref{eq:toyexp} can not hold, but is instead replaced by
\begin{korollar}
	For local $\toy \in \chars{H}{\alg}$, the \emph{beta functional} \cite{CK:RH2}
	\begin{equation}\label{eq:MS-beta-def}
		\beta
		\defas
			\reg \cdot \MSZ[][\convolution -1] \circ (S\convolution Y)
		= -\Res \circ \MSZ \circ Y
		= -\Res \circ \MSZ \circ (S\convolution Y)
		\in \infchars{H}{\K}
	\end{equation}
	completely dictates the scale dependence of the physical limit $\MSRR$ through
	\begin{equation}\label{eq:MS-scale-dependence}
		\MSRR[s]
		= \lim_{\reg \rightarrow 0} (\toy \circ \gradAut_{-\reg\scalelog})_+
		= \exp_{\convolution}(-\scalelog\beta) \convolution \left( \counit \circ \MSRR \right).
	\end{equation}
\end{korollar}
\begin{proof}
	As $\MSZ[][\convolution-1]$ is local by lemma \ref{satz:msz-inv-local}, proposition \ref{satz:finiteness-algebraic} allows to invoke \ref{satz:finiteness=locality} which proves the finiteness of \eqref{eq:MS-beta-def}:
	$\im \left(\MSZ[][\convolution-1] \circ (S\convolution Y)\right) \subseteq \reg^{-1} \K[[\reg]] \cap \alg_- = \K \cdot \reg^{-1}$,
	\begin{equation}
		\beta
		= \lim_{\reg \rightarrow 0} \left[ \reg \cdot \MSZ[][\convolution - 1] \circ (S \convolution Y) \right]
		= \reg \cdot \MSZ[][\convolution - 1] \circ (S \convolution Y)
		\in \infchars{H}{\K}
	\end{equation}
	converges and \eqref{eq:toyexp} applies to give $\lim_{\reg\rightarrow 0} \MSZ[][\convolution-1] \circ (S \convolution \gradAut_{-\reg\scalelog}) = \exp_{\convolution} (-\scalelog\beta)$. Inserting this into the Birkhoff decomposition (here we set $\MSR = \MSR[\rp] = \restrict{\MSR}{\scalelog=0}$)
	\begin{align*}
		\MSR[s]
		&= \left( \toy \circ \gradAut_{-\reg\scalelog} \right)_- 
				\convolution 
				\left( \toy \circ \gradAut_{-\reg\scalelog} \right)
		= \MSZ \convolution 
			\left[ (\MSZ[][\convolution-1] \convolution \MSR) \circ \gradAut_{-\reg\scalelog} \right] \\
		&= \left[ \MSR[][\convolution-1] \circ (S \convolution \gradAut_{-\reg\scalelog} )\right] 
				\convolution 
				\left( \MSR \circ \gradAut_{-\reg\scalelog} \right)
	\end{align*}
	and exploiting $\lim_{\reg\rightarrow 0}\left( \MSR \circ \gradAut_{-\reg\scalelog} \right) = \lim_{\reg\rightarrow 0} \MSR = \restrict{\MSRR}{\scalelog=0}$ shows \eqref{eq:MS-scale-dependence}.
	The relations in \eqref{eq:MS-beta-def} follow from $S(\f) = -\f \mod (\ker \counit)^2$, $S \convolution Y (\f) = Y(\f) \mod (\ker\counit)^2$ and the fact that $\Res \circ \MSZ \in \infchars{H}{\K}$ vanishes on products, because for any $\f,\f' \in \ker \counit$
	\begin{equation*}
		\MSZ (\f \cdot \f')
		= \MSZ (\f) \cdot \MSZ (\f')
		\in \alg_-^2
		= \reg^{-2} \K[[\reg]]
	\end{equation*}
	has no pole of first order.
\end{proof}
\begin{lemma}
	Let $\toy\in\chars{H}{\alg}$ be local, then $\MSZ[][\convolution - 1] \in \chars{H}{\alg}$ is local as well.
	\label{satz:msz-inv-local}
\end{lemma}
\begin{proof}
	As in \cite{Manchon}, from $\MSR, \MSR[s] \in \chars{H}{\alg_+}$ we deduce that
	\begin{equation*}
		\MSZ[][\convolution - 1] \circ \gradAut_{-\reg\scalelog}
		= \left( \toy \convolution \MSR[][\convolution -1] \right) \circ \gradAut_{-\reg\scalelog}
		= \toy[s] \convolution \left(\MSR[][\convolution -1] \circ \gradAut_{-\reg\scalelog} \right)
		= \MSZ[][\convolution - 1] \convolution \MSR[s] \convolution \left( \MSR[][\convolution - 1] \circ \gradAut_{-\reg\scalelog} \right)
	\end{equation*}
	is a Birkhoff decomposition and read off $\left( \MSZ[][\convolution - 1] \circ \gradAut_{-\reg\scalelog} \right)_- = \MSZ$ by uniqueness.
\end{proof}
In the minimal subtraction scheme we can rephrase the renormalization group \eqref{eq:MS-scale-dependence} as expressing all $\reg$-poles of the counterterms $\MSZ: H \rightarrow \K[\reg^{-1}]$ in terms of the first order poles only \cite{CK:RH2}.
	This comes about as $\beta$ captures all information on the character $\MSZ[][\convolution-1] \in \infchars{H}{\alg}$ since $\im (S\convolution Y)$ generates the full Hopf algebra by \ref{satz:S*Y-generates-H}. 
	For clarity we shall demonstrate this in
\begin{beispiel}\label{ex:MS-RGE-poles}
	First we take example \ref{ex:MS-R} to read off the counterterms
	\begin{align*}
		\MSZ[][\convolution - 1] \left( \tree{+-} \right)
		&= \MSZ \left( - \tree{+-} \right)
		= -\Rms \left[ - e^{-\reg\scalelog} F(\reg) \right]
		= \frac{\coeff{-1}}{\reg} \\
		\MSZ[][\convolution - 1] \left( \tree{++--} \right)
		&=\MSZ \left( -\tree{++--} + \tree{+-}\tree{+-} \right)
		= -\frac{\coeff[2]{-1}}{2\reg^2} + \frac{\coeff{-1}\coeff{0}}{2\reg} + \left[ \frac{\coeff{-1}}{\reg} \right]^2
		= \frac{\coeff[2]{-1}}{2\reg^2} + \frac{\coeff{-1}\coeff{0}}{2\reg}
	\end{align*}
	and apply \eqref{eq:MS-beta-def} the get hold of the residues $\beta\left( \tree{+-} \right) \urel{\eqref{eq:MS-beta-def}} \reg \cdot \MSZ[][\convolution-1] \left( \tree{+-} \right) = \coeff{-1}$ as well as
	\begin{equation*}
		\beta\left( \tree{++--} \right) 
		\urel{\eqref{eq:MS-beta-def}}
		\reg \cdot \MSZ[][\convolution -1] \left( 2\tree{++--} - \tree{+-}\tree{+-} \right)
		= \coeff{-1}\coeff{0}.
	\end{equation*}
	Observe how for $S\convolution Y \left( \tree{++--} \right) = 2\tree{++--} - \tree{+-}\tree{+-}$ we indeed obtained only a first order pole in $\MSZ[][\convolution-1]$, contrary to $\tree{++--}$ itself entailing a second order pole as well.
	This one we can now predict using \eqref{eq:Rtilde-inverse} applied to $\MSZ[][\convolution-1] \circ (S\convolution Y) = \frac{\beta}{\reg}$:
	\begin{equation*}
		\MSZ[][\convolution-1] \left( \tree{++--} \right)
		= \left\{ e + \frac{\beta \circ Y^{-1}}{\reg} + \frac{\left[ (\beta \circ Y^{-1}) \convolution \beta \right] \circ Y^{-1}}{\reg^2} \right\} \left( \tree{++--} \right)
		= \frac{\beta\left( \frac{1}{2} \tree{++--} \right)}{\reg} + \frac{\left[\beta\left( \tree{+-} \right)\right]^2}{2\reg^2}
		= \frac{\coeff{-1}\coeff{0}}{2\reg} + \frac{\coeff[2]{-1}}{2\reg^2}.
	\end{equation*}
\end{beispiel}

\subsection{Relating \texorpdfstring{$\momsch{\rp}$}{\momscheme} with \texorpdfstring{$\Rms$}{minimal subtraction}}
	Though both schemes seem so different, already \cite{BroadhurstKreimer:Auto} exploited their relationship given in
\begin{lemma}\label{satz:phiR-MSR}
	For local $\toy \in \chars{H}{\alg}$, the scale dependence of $\MSR[s]$ (in the $\Rms$ scheme) is dictated by $\toyR[s]$ (\momscheme) through
	\begin{equation}
		\MSR[s]
		= \left( \momsch{\rp} \circ \MSR[s] \right) \convolution \toyR[s].
		\label{eq:phiR-MSR}
	\end{equation}
\end{lemma}
\begin{proof}
	Locality of the counterterms $\MSZ$ implies $\momsch{\rp} \circ \MSZ = \MSZ$, hence
	\begin{align*}
		&\left( \momsch{\rp} \circ \MSR[s] \right) \convolution \toyR[s]
		= \left[ \momsch{\rp} \circ \left( \MSZ \convolution \toy[s] \right) \right]
			\convolution \left( \momsch{\rp} \circ \toy[s] \right)^{\convolution -1} \convolution \toy[s] \\
		&= \left[ \momsch{\rp} \circ \left( \MSZ \convolution \toy[s] \convolution \toy[s][\convolution-1] \right) \right] \convolution \toy[s]
		= \left( \momsch{\rp} \circ \MSZ \right) \convolution \toy[s]
		= \MSR[s]. \qedhere
	\end{align*}
\end{proof}
Note how $\momsch{\rp} \circ \MSR[s] = \restrict{\MSR}{\scalelog=0}$ reduces to the constants $\counit \circ \MSRR \in \chars{H}{\K}$ in the physical limit.
There \eqref{eq:phiR-MSR} takes the form of
\begin{korollar}
	The characters $\phiR, \MSRR:\ H_R\rightarrow \K[\x]$ fulfill the relations
	\begin{equation}\label{eq:MSR-coprod}
		\MSRR
		= \left( \counit \circ \MSRR \right) \convolution \phiR
		,\quad\text{equivalently}\quad
		\Delta \circ \MSRR 
		= \left( \MSRR \tp \phiR \right) \circ \Delta.
	\end{equation}
\end{korollar}
\begin{beispiel}
	After reading off the constants $\counit \circ \MSRR \left( \tree{+-} \right) = \coeff{0}$ and $\counit \circ \MSRR \left( \tree{++--} \right) = \coeff[2]{0} + \frac{3}{2}\coeff{-1}\coeff{1}$, we can verify \eqref{eq:MSR-coprod} against the examples \ref{ex:MS-R} making use of \ref{ex:renormalized-universal}:
	\begin{align*}
		\MSRR \left( \tree{+-} \right)
		&= \counit \circ \MSRR \left( \tree{+-} \right) + \phiR\left( \tree{+-} \right)
		= \coeff{0} - \coeff{-1}\scalelog \\
		\MSRR \left( \tree{++--} \right)
		&= \counit \circ \MSRR \left( \tree{++--} \right) + \left[ \counit \circ \MSRR\left( \tree{+-} \right) \right] \cdot \phiR\left( \tree{+-} \right) + \phiR\left( \tree{++--} \right) \\
		&= \coeff[2]{0} + \frac{3}{2}\coeff{-1}\coeff{1} - \coeff{0}\coeff{-1}\scalelog + \frac{\coeff[2]{-1}}{2} \scalelog^2 - \coeff{0}\coeff{-1}\scalelog
		=  \frac{\coeff[2]{-1}}{2} \scalelog^2 - 2\coeff{0}\coeff{-1}\scalelog + \coeff[2]{0} + \frac{3}{2}\coeff{-1}\coeff{1}.
	\end{align*}
\end{beispiel}
\begin{korollar}\label{satz:beta-toylog-conjugation}
	Inserting both \eqref{eq:MS-scale-dependence} and \eqref{eq:toyexp} into \eqref{eq:MSR-coprod} reveals
	\begin{equation}\label{eq:beta-toylog-conjugation}
		\beta \convolution \left( \counit \circ \MSRR \right)
		= \left( \counit \circ \MSRR \right) \convolution \toylog.
	\end{equation}
\end{korollar}
	Hence $\beta$ and $\toylog$ differ only by conjugation with the character $\counit\circ\MSRR \in \chars{H}{\K}$ and therefore in particular agree on any cocommutative elements of $H$.
\begin{beispiel}
	For the rooted trees $H_R$, the cocommutative elements include the \emph{ladders} $B_+^n (\1)$. In the examples \ref{ex:MS-RGE-poles} and \ref{ex:toylog-mellin-recursion} we explicitly checked the first two cases of $n \in \set{1,2}$: $\beta\left( \tree{+-} \right) = \coeff{-1} = \toylog\left( \tree{+-} \right)$ as well as $\beta\left( \tree{++--} \right) = \coeff{-1}\coeff{0} = \toylog\left( \tree{++--} \right)$.
\end{beispiel}

\hide{
	\begin{lemma}\label{satz:lokal-convolution}
	Let $\toy \in \chars{H_R}{\alg}$ be finite and $\psi: H_R \rightarrow \K[[\reg]]$ holomorphic and linear. Then $\toy \convolution \psi$ has finite physical limit (setting $\psi_0 \defas \restrict{\psi}{\reg=0}$)
	\begin{equation}
		(\phi\convolution \psi)_{R}
		= \psi_0^{\convolution -1} \convolution \phiR \convolution \psi_0.
		\label{eq:lokal-convolution}
	\end{equation}
\end{lemma}
\begin{proof}
	Take the limit $\reg \rightarrow 0$ of
	\begin{align*}
		(\toy \convolution \psi) \circ (S\convolution \gradAut_{-\reg\scalelog})
		&= m^3 \circ (\toy \tp \psi \tp \toy \tp \psi) \circ (\Delta S \tp \Delta \gradAut_{-\reg\scalelog}) \circ \Delta \\
		&= m^3 \circ (\psi \tp \toy \tp \toy \tp \psi) \circ (S\tp S \tp \gradAut_{-\reg\scalelog} \tp \gradAut_{-\reg\scalelog}) \circ \Delta^3 \\
		&= \psi^{\convolution -1} \convolution \left[ \toy \circ (S\convolution\gradAut_{-\reg\scalelog}) \right] \convolution \left( \psi \circ \gradAut_{-\reg\scalelog} \right). \qedhere
	\end{align*}
\end{proof}
\begin{lemma}
	Let $\toy\in\chars{H_R}{\alg}$ be local and $\psi: H_R \rightarrow \reg\K[[\reg]]$ linear, then $\toy \circ \gradAut_{\psi(\reg)}$ is local with physical limit $\left( \toy \circ \gradAut_{\psi(\reg)} \right)_R = \toy_R$.
\end{lemma}
\begin{proof}
	Immediate from 
	$
		\toy \circ \gradAut_{\psi(\reg)} \circ (S\convolution\gradAut_{-\reg\scalelog})
		= \toy \circ (S \convolution \gradAut_{-\reg\scalelog}) \circ \gradAut_{\psi(\reg)}
	$.
\end{proof}
}

\section{Dyson-Schwinger equations and correlation functions}
\label{sec:DSE}

Until now we considered the renormalized Feynman rules $\phiR$ on their own, but these form only one ingredient to {\qft}. The counterpart is the \emph{perturbation series} $X(\coupling)$ they are being applied to.
\begin{beispiel}[\cite{Panzer:Master}]\label{ex:dse-yukawa}
	In \emph{Yukawa theory}, the propagation of a fermion is the superposition of infinitely many possible interactions with a scalar field each of which is represented by a Feynman diagram. Among these are contributions like
\begin{multline*}
	 \1
		-\Graph{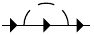} \coupling
		-\Graph{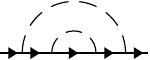} \coupling^2
		-\left(
				\Graph{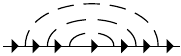} + \Graph{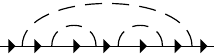}
			\right) \coupling^3 \\
		-\left(
				\Graph[0.7]{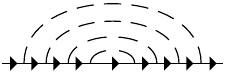} + \Graph[0.7]{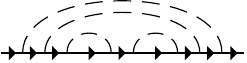}
				+\Graph[0.7]{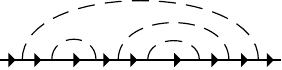} + \Graph[0.7]{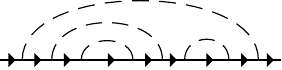} + \Graph[0.7]{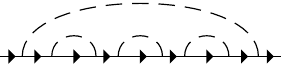}
			\right) \coupling^4
		+\bigo{\coupling^5}.
\end{multline*}
The Feynman rules map each graph to an individual amplitude, but physically these are not distinguishable and need to be all summed up. Further, a \emph{coupling constant} $g$ takes the strength of an interaction into account.
\end{beispiel}
\begin{definition}\label{def:correlation}
	A \emph{perturbation series} is a formal power series
	\begin{equation}\label{eq:perturbation-series}
		X(\coupling) 
		= \sum_{n\in\N_0} 
				x_n
				\coupling^n
		\in H_R[[\coupling]]
		\quad\text{with}\quad
		x_0 = \1
	\end{equation}
	taking values in the Hopf algebra $H_R$ of rooted trees and indexed by the \emph{coupling constant} $\coupling$.
	Evaluation of the renormalized Feynman rules $\phiR \in \chars{H_R}{\K[\scalelog]}$ on $X(\coupling)$ delivers the \emph{correlation function}
	\begin{equation}\label{eq:correlation}
		G (\coupling)
		\defas \phiR \circ X(\coupling)
		= \sum_{n\in\N_0} \phiR (x_n) \coupling^n
		\in \K[\scalelog] [[\coupling]],
	\end{equation}
	while the \emph{physical anomalous dimension} is
	$	\widetilde{\toylog}(\coupling)
		\defas \toylog \circ X(\coupling)
		\urel{\eqref{eq:toylog}} \restrict{-\partial_{\scalelog}}{0} G(\coupling)
		\in {\K}[[\coupling]]
	$.
\end{definition}
The crucial property of perturbation series is the possibility of insertions: In the above example, we started with the \emph{primitive}\footnote{That means it is free of subdivergences.} graph $\Graph{+-}$ and iteratively inserted it as a subdivergence into itself.
Though the full perturbation series contains many more graphs, this illustrates how $X(\coupling)$ may efficiently be described by means of recursive insertions. Those are represented by Hochschild-1-cocycles motivating
\begin{definition}\label{def:dse}
	To a parameter $\powdep\in\K$ and a family of cocycles $B_{\cdot}\!: \N\rightarrow \HZ[1](H_R)$ we associate the \emph{combinatorial Dyson-Schwinger equation}\footnote{
	As $x_0=\1$, for arbitrary $p$ we define $\left[ X(\coupling) \right]^{p} \defas \sum_{n\in\N_0} \binom{p}{n} \left[ X(\coupling)-\1 \right]^n \in H_R[[\coupling]]$.}
	\begin{equation}
		X(\coupling)
		= \1 + \sum_{n\in\N} \coupling^n B_n \left( X^{1+n\powdep}(\coupling) \right).
		\label{eq:dse}
	\end{equation}
\end{definition}
This type of equations is folklore in physics, but had not been cast into its pure algebraic form before \cite{BergbauerKreimer}. Referring to \cite{Foissy:DSE} we recall the main results in
\begin{lemma}\label{satz:dse-coprod}
	As \emph{perturbation series} \eqref{eq:perturbation-series}, the equation \eqref{eq:dse} allows a unique solution which is determined recursively by
	\begin{equation}
		x_{k} 
		= \sum_{0 \leq m+n\leq k}
				\binom{1+\powdep n}{m}
				B_n \left( 
					\sum_{\substack{
						i_1+\ldots+i_m+n = k \\
						i_1,\ldots,i_m \geq 1
					}}
				x_{i_1} \cdots x_{i_m}.
			\right)
	\end{equation}
	Most importantly, these coefficients generate a Hopf subalgebra $H_X \defas \langle\setexp{x_n}{n\in\N}\rangle$ (isomorphic to the F\`{a}a di Bruno Hopf algebra when $\powdep \neq 0$). Explicitly we find
	\begin{equation}\label{eq:dse-coprod}
		\Delta X(\coupling)
		= \sum_{n\in\N_0} \left[ X(\coupling) \right]^{1+n\powdep} \tp \coupling^n x_n
		\in (H_R\tp H_R) [[\coupling]].
	\end{equation}
\end{lemma}
We learn that the solution of \eqref{eq:dse} has a very special property: The coproduct $\Delta x_n \in H_X \tp H_X$ can be expressed by coefficients $x_{\cdot}$ alone, with \eqref{eq:dse-coprod} serving an explicit formula. 
Before we exploit this information on $X$ let us give some examples.
\begin{beispiel}
	When we set $\powdep=0$, $\Delta X(\coupling) = X(\coupling) \tp X(\coupling)$ is group-like such that $\Delta x_n = \sum_{i+j=n} x_i \tp x_j$. 
	The corresponding Dyson-Schwinger equation $X(\coupling) = \1 + B_+\left( X(\coupling) \right)$ is linear and generates the cocommutative ladders $x_n = B_+^n(\1)$.
\end{beispiel}
Recall that we take the trees in $H_R$ as substitute for Feynman graphs, each node representing an insertion into some other graph.
\begin{beispiel}\label{ex:DSE-propagator}
	In \cite{Kreimer:ExactDSE,Panzer:Master} we find the equation $X(\coupling) = \1 - \coupling B_+ \left( \frac{1}{X(\coupling)} \right)$ featuring $\powdep=-2$ which corresponds to the propagator example \ref{ex:dse-yukawa}. The solution sums all trees with the factor counting the number of distinct ordered embeddings:
\begin{equation*}\begin{split}
	X(\coupling)\in
	&\	\1
		- \tree{+-}\ \coupling
		- \tree{++--}\ \coupling^2
		- \left(
				\tree{+++---} + \tree{++-+--}
			\right) \coupling^3
		- \left(
					\tree{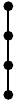} 
				+	\tree{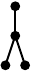}
				+2\tree{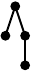}
				+	\tree{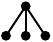} 
			\right) \coupling^4 \\
	&	- \left(
					\tree{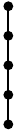} + \tree{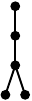} 
				+2\tree{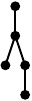} + \tree{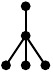}
				+ \tree{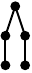} + 2\tree{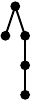} 
				+2\tree{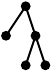} + 3\tree{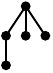} 
				+ \tree{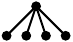}
			\right) \coupling^5
		+ \coupling^6 H_R[[\coupling]].
\end{split}\end{equation*}
	The first factor of two arises from the different embeddings $\tree{++-++---}$ and $\tree{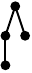}$ and correctly accounts for the fact that these two shall represent different graphs (though they are the same elements in $H_R$): $\Graph[0.5]{++-++---}$ and $\Graph[0.5]{+++--+--}$.
\end{beispiel}
\begin{beispiel}\label{ex:intrules-correlation}
	Let us consider the tree factorial Feynman rules $\intrules$ from \eqref{eq:int-rules} applied to the above series. The anomalous dimension
	$
		\widetilde{\toylog}(\coupling) 
		\urel{\ref{ex:intrules-exp}} - Z_{\tree{+-}} \circ X(\coupling)
		= \coupling
	$ is linear,
	while the first terms of the correlation function become
\begin{align*}
	G(\coupling)
	&= 1
		-\frac{(\coupling \scalelog)}{\tree{+-}\,!}
		-\frac{(\coupling \scalelog)^2}{\tree{++--}\,!}
		-\frac{(\coupling \scalelog)^3}{\tree{+++---}\,!} 
		-\frac{(\coupling \scalelog)^3}{\tree{++-+--}\,!}
		- \ldots
	 = 1
		-\coupling\scalelog
		-\frac{1}{2} (\coupling\scalelog)^2
		-\frac{1}{2} (\coupling\scalelog)^3 
		+\bigo{{(\coupling\scalelog)}^4}.
\end{align*}
\end{beispiel}

\subsection{Propagator coupling duality}
\label{sec:propagator-coupling}

The Hopf subalgebra of the perturbation series allows to calculate convolutions in
\begin{lemma}\label{lemma:perturbation-convolutions}
	Let $\psi\in\infchars{H_R}{\alg}$ denote an infinitesimal character, $\Psi\in\chars{H_R}{\alg}$ a character and $\lambda\in\Hom(H_R,\alg)$ a linear map. Then we obtain
	\begin{align}
		(\Psi \convolution \lambda) \circ X(\coupling)
		=& \left[ \Psi \circ X(\coupling)\right] 
				\cdot \lambda\circ X\left(
					\coupling 
					\left[ \Psi \circ X(\coupling) \right]^{\powdep}
				\right)
		\label{eq:perturbation-convolution-char}\\
		\defas&
				\left[ \Psi \circ X(\coupling)\right] 
				\cdot \sum_{n\in\N_0} \lambda(x_n) \cdot \left(
					\coupling 
					\left[ \Psi \circ X(\coupling) \right]^{\powdep}
				\right)^n
		\in{\alg}[[\coupling]] \nonumber\\
		(\psi \convolution \lambda) \circ X(\coupling)
		=&	\left[ \psi\circ X(\coupling) \right] \cdot
				\left( \id + \powdep \coupling \partial_{\coupling} \right)
				\left[ \lambda \circ X(\coupling) \right]
		\in{\alg} [[\coupling]].
		\label{eq:perturbation-convolution-inf}
	\end{align}
\end{lemma}
\begin{proof}
	These are immediate consequences of lemma \ref{satz:dse-coprod}, for \eqref{eq:perturbation-convolution-inf} consider
	\begin{equation*}
				\psi \left( \left[ X(\coupling) \right]^{1+n\powdep} \right) \cdot
					\coupling^n
		=		\sum_{i\in\N_0}
				\tbinom{1+n\powdep}{i}
				\psi \left( \left[ X(\coupling) - \1 \right]^{i} \right)
					\coupling^n
		 =	\psi\left( X(\coupling) - \1 \right)
				\cdot
				(1+n\powdep) 
					\coupling^n 
		. \qedhere
	\end{equation*}
\end{proof}
By combining \eqref{eq:dse-coprod} with the renormalization group equation we can calculate the correlation function out of the knowledge of $\widetilde{\toylog}$ only:
\begin{beispiel}\label{ex:intrules-propagator}
	Continuing example \ref{ex:DSE-propagator}, we can calculate the $\propto \scalelog^2$-term $Z_{\tree{+-}}^{\convolution 2} \left(X(\coupling)\right) = -\coupling(1-2\coupling\partial_{\coupling})(-\coupling) = -\coupling^2$ and all further convolution products
\begin{equation*}
	Z_{\tree{+-}}^{\convolution n+1} \left( X(\coupling) \right)
	\urel{\eqref{eq:perturbation-convolution-inf}}
	-\coupling^{n+1} (2n-1)(2n-3)\cdots(1)
	= -\coupling^{n+1} \frac{(2n)!}{2^n n!},
\end{equation*}
proving $\intrules(x_{n+1}) = -2^{-n} C_n \scalelog^{n+1}$ with the \emph{Catalan numbers} $C_n$ already noted in \cite{Kreimer:NonlinearDSE}. Combining their known generating function
	$	2\coupling \sum_{n\in\N_0} \coupling^n C_n %
		= 1-\sqrt{1-4\coupling}	$
with $\intrules = \exp_{\convolution}\left(-\scalelog Z_{\tree{+-}} \right)$ allows us to completely determine the correlation function as $G(\coupling) = \sqrt{1-2\coupling\scalelog}$.
\end{beispiel}

\begin{korollar}
	As $\phiR$ is a morphism of Hopf algebras by \ref{satz:toyphy-hopfmor}, for any $\scalelog,\scalelog'\in\K$ we can factorize the correlation function at $\scalelog+\scalelog'$ in two different ways
	\begin{equation}
		G_{\scalelog+\scalelog'}(\coupling)
		= (\phiR[\scalelog]	\convolution	\phiR[\scalelog']) \circ X(\coupling)
		\urel{\eqref{eq:perturbation-convolution-char}}
		G_{\scalelog} (\coupling) \cdot G_{\scalelog'}\left[ \coupling G_{\scalelog}^{\powdep}(\coupling) \right]
		= G_{\scalelog'} (\coupling) \cdot G_{\scalelog}\left[ \coupling G_{\scalelog'}^{\powdep}(\coupling) \right].
		\label{eq:propagator-coupling}
	\end{equation}
\end{korollar}
At this point we like to briefly highlight the non-analytic nature of perturbative {\qft}.
The correlation function $G(\coupling)$ at $\coupling \neq 0$ \eqref{eq:correlation} is a physical object that can in principle be measured through experiment.
It is only by the nature of the perturbative method we apply that we are merely able to calculate the formal series expansion \eqref{eq:correlation} of $G(\coupling)$ around $\coupling=0$ with the help of the Feynman rules.
The main issue is that in the interesting cases, the function $G(\coupling)$ is not analytic at this point and the series \eqref{eq:correlation} has zero radius of convergence\footnote{Hence the example \ref{ex:intrules-propagator} is still far away from {\qft} as its correlation function is analytic at $g\rightarrow 0$.}.

However, in this perturbative approach we just deduced the functional equations \eqref{eq:propagator-coupling} for the formal series.
Therefore it is natural to impose these on the true correlation functions, such that we gain a non-perturbative handle on {\qft}.
We will continue to stress similar examples in this section.

But first observe how \eqref{eq:propagator-coupling} takes the infinitesimal form
\begin{korollar}
	With the help of $-\frac{\dd}{\dd\x}\ \phiR = \toylog \convolution \phiR = \phiR \convolution \toylog$ or by differentiating \eqref{eq:propagator-coupling} with respect to $\scalelog'$ at zero we find the differential equations
	\begin{equation}\label{eq:propagator-coupling-diff}
		G_{\scalelog}(\coupling) \cdot \widetilde{\toylog}\left[ \coupling G_{\scalelog}^{\powdep}(\coupling) \right]
		\urel{\eqref{eq:perturbation-convolution-char}}
		-\partial_{\scalelog} G_{\scalelog}(\coupling)
		\urel{\eqref{eq:perturbation-convolution-inf}}
		\widetilde{\toylog}(\coupling) \cdot \left( 1 + \powdep \coupling \partial_{\coupling} \right) G_{\scalelog}(\coupling).
	\end{equation}
\end{korollar}
The first of these equations generalizes the \emph{propagator coupling duality} observed in \cite{Kreimer:ExactDSE,Kreimer:NonlinearDSE}. For any fixed coupling $\coupling$, it expresses the correlation function as the solution of the first order ordinary differential equation
\begin{equation}
	-\frac{\dd}{\dd \scalelog} \ln G_{\scalelog} (\coupling)
	= \widetilde{\toylog} \left[ \coupling e^{\powdep \ln G_{\scalelog}(\coupling)} \right]
	\quad\text{with initial condition}\quad
	\ln G_0 (\coupling) = 0.
	\label{eq:propagator-coupling-ode}
\end{equation}
Note how this equation reconstructs $G_{\scalelog}(\coupling)$ completely only from the input $\widetilde{\toylog}(\coupling)$. This demonstrates the power of the renormalization group: Though $G$ depends on $\coupling$ and $\scalelog$, after imposing \eqref{eq:rge} only a one-dimensional degree of freedom is left.
As before, \eqref{eq:propagator-coupling-ode} serves a non-perturbative relation and need not be restrained to the perturbative series alone.
\begin{beispiel}\label{ex:leading-log}
	The \emph{leading-$\log$} expansion takes only the highest power of $\scalelog$ in each $\coupling$-order, so $\widetilde{\toylog} (\coupling) = c \coupling^n$ is a monomial for some $c\in\K$, $n\in\N$ (otherwise different $\coupling$-powers would mix for a given order in $\scalelog$). In this case \eqref{eq:propagator-coupling-ode} integrates to
	\begin{equation}\label{eq:leading-log}
		G_{\text{leading-$\log$}} (\coupling)
		= \Big[ 1 + c n \powdep \scalelog \coupling^n \Big]^{-\frac{1}{n\powdep}}.
	\end{equation}
As a special case we recover example \ref{ex:intrules-correlation} for $n=c=1$ and $\powdep = -2$.
\end{beispiel}
\begin{beispiel}
	For the linear Dyson-Schwinger equation $\powdep=0$, \eqref{eq:propagator-coupling} states $G_{\scalelog+\scalelog'}(\coupling) = G_{\scalelog}(\coupling) \cdot G_{\scalelog'}(\coupling)$ which is solved by the \emph{scaling solution} 
$
	G_{\scalelog}(\coupling) 
	= e^{-\scalelog\tilde{\toylog}(\coupling)}
$
of \eqref{eq:propagator-coupling-ode}, well-known from \cite{Kreimer:linear}.
\end{beispiel}
\begin{beispiel}\label{ex:vertex-insertion}
	The physical situation of vertex insertions as in \cite{BKW:NextToLadder} corresponds to $\powdep=1$ and
$
	G_{\scalelog+\scalelog'}(\coupling) 
	=	G_{\scalelog'}(\coupling) \cdot G_{\scalelog} \big[ \widetilde{G}_{\scalelog'}(\coupling) \big]
$
can be interpreted as the running of the coupling constant $\widetilde{G} \defas \coupling \cdot G$: A change in scale by $\scalelog'$ is (up to a multiplicative constant) equivalent to replacing the coupling $\coupling$ by $\widetilde{G}_{\scalelog'}(\coupling)$.
\end{beispiel}

\subsection{Running coupling}\label{sec:running-coupling}
The idea of this last example \ref{ex:vertex-insertion} leads us to another form of the renormalization group equation, common to the physics literature like (7.3.15) and (7.3.21) in \cite{Collins}. 
We introduce the \emph{$\beta$-function}\footnote{This should not be confused with the $\beta$-functional of \eqref{eq:MS-beta-def}, though both are related.} $\beta(\coupling) \defas -\powdep\coupling\widetilde{\toylog}(\coupling)$ and the \emph{running coupling} $\coupling(\rp)$ as the solution of
\begin{equation}\label{eq:running-coupling}
	\rp\frac{\dd}{\dd\rp} \coupling(\rp)
	= \beta \big( \coupling(\rp) \big),
	\quad\text{so}\quad
	\rp\frac{\dd}{\dd\rp} G \left( \coupling(\rp),\ln\tfrac{s}{\rp} \right)
	\urel{\eqref{eq:propagator-coupling-diff}}
	\widetilde{\toylog}\big( \coupling(\rp) \big) G \left( \coupling(\rp),\ln\tfrac{s}{\rp} \right).
\end{equation}
Integration over $\rp$ results in a relation of the correlation functions corresponding to different choices of the renormalization point $\rp$:
\begin{equation*}
	G\left( \coupling(\rp_2),\ln\tfrac{s}{\rp_2} \right) 
	= G\left( \coupling(\rp_1),\ln\tfrac{s}{\rp_1} \right) \cdot
		\exp \left[ 
			\int_{\rp_1}^{\rp_2} \widetilde{\toylog}\big( \coupling(\rp) \big) \tfrac{\dd\rp}{\rp}
		\right]
	\urel{\eqref{eq:running-coupling}}
		G\left( \coupling(\rp_1),\ln\tfrac{s}{\rp_1} \right) \cdot
		\left[ 
			\frac{\coupling(\rp_2)}{\coupling(\rp_1)}
		\right]^{-\frac{1}{\powdep}}.
\end{equation*}
This result is important from a conceptual point of view: To achieve renormalization, we introduced a parameter $\rp$ that is completely arbitrary, yet the shape of the correlation function is a measurable quantity wherefore it clearly has to be insensitive to the choice of $\rp$.

Indeed we see that a change from $\rp_1$ to $\rp_2$ affects the correlation function only by a constant overall factor and a redefinition of the coupling constant (which itself is a parameter of the theory). Hence the physical content of $G$ is left invariant of the choice of renormalization point.

Similar to \eqref{eq:propagator-coupling-ode} we can express $G$ through a differential equation involving the running coupling $\coupling(s)$ after choosing $\rp_1=s$:
\begin{equation}\label{eq:rge-solved}
	G_{\scalelog}( \coupling)
	= \left[
			\frac{\coupling}{\coupling(s)}
		\right]^{-\frac{1}{\powdep}},
	\quad\text{with $\coupling(s)$ subject to}\quad
	\scalelog
	= \ln\frac{s}{\rp} 
	= \int_{\coupling}^{\coupling(s)} \frac{\dd\coupling'}{\beta(\coupling')}.
\end{equation}

\subsection{Relation to Mellin transforms}
So far we exploited the renormalization group equation \eqref{eq:toyexp} and the Hopf subalgebra \eqref{eq:dse-coprod} of perturbation series.
Now we like to take the special structure \eqref{eq:phiR-universal} of the Feynman rules from \eqref{eq:renormiert-def} into account. As different cocycles\footnote{This is most easily considered in the Hopf algebra of decorated rooted trees (section \ref{sec:decorations}) where each $B_n$ inserts into a node of decoration $n$.} represent different Feynman integrals, we allow for a family $F_{\cdot}: \N \rightarrow \reg^{-1} \K [[\reg]]$ of Mellin transforms such that $\phiR \circ B_n = P \circ F_n(-\partial_{\scalelog}) \circ \phiR$. Applying this to \eqref{eq:dse} results in
\begin{equation*}
	G_{\scalelog}(\coupling)
	\urel{\eqref{eq:dse}}
		1 + \sum_{n\in\N} \coupling^n \phiR \circ B_n \left(X(\coupling)^{1+n\powdep}\right)
	\urel{\eqref{eq:phiR-universal}}
	1+ P \circ \sum_{n\in\N} \coupling^n F_n\left( -\partial_{\scalelog} \right) G_{\scalelog}(\coupling)^{1+n\powdep}
\end{equation*}
and taking a derivative brings us to the differential equation of
\begin{korollar}
	The power series $G_{\scalelog}(\coupling)\in{\K}[\scalelog][[\coupling]]$ is fully determined by
	\begin{equation}\label{eq:correlation-mellin-de}
		\partial_{-\scalelog} G_{\scalelog} \left( \coupling \right)
		\urel{\eqref{eq:phiR-universal}} 
		\sum_{n\in\N} \coupling^n \left[\reg F_n(\reg)\right]_{\reg = - \partial_{\scalelog}} \left( {G_{\scalelog}(\coupling)^{1+n\powdep}} \right)
		\quad\text{and}\quad
		G_{\scalelog}(0) = 1.
	\end{equation}
\end{korollar}
Though \eqref{eq:correlation-mellin-de} is always well defined for the formal power series (since the growing powers of $\coupling$ allow only finitely many contributions in each order), the differential operator $[\reg F(\reg)]_{\reg = -\partial_{\scalelog}}$ can be of infinite order which might hinder a non-perturbative interpretation of this equation.
However we can proceed in a couple of interesting cases, allowing us to construct the full correlation function or anomalous dimension (and not just individual terms of the perturbation series):
\begin{beispiel}
	Consider a single cocycle $F_k(\reg) = F(\reg) \delta_{k,n}$ and choose $F(\reg)=\frac{\coeff{-1}}{\reg}$. Then \eqref{eq:correlation-mellin-de} reproduces the leading-log example \eqref{eq:leading-log} as it becomes
	\begin{equation}
		\partial_{-\scalelog} G_{\scalelog}(\coupling) 
		= \coupling^n \coeff{-1} G_{\scalelog}(\coupling)^{1+n\powdep}.
	\end{equation}
\end{beispiel}
More generally, for rational $F(\reg) = \frac{p(\reg)}{q(\reg)} \in \K(\reg)$  with polynomials $p(\reg), q(\reg) \in \K[\reg]$ we can apply $q(-\partial_{\scalelog})$ on both sides of \eqref{eq:correlation-mellin-de} resulting in a finite order differential equation
\begin{equation}
	q(-\partial_{\scalelog}) G_{\scalelog}(\coupling) 
	= \coupling^n p(-\partial_{\scalelog}) G_{\scalelog}(\coupling)^{1+n\powdep}.
\end{equation}
Enjoying this situation we can directly interpret it non-perturbatively (extending the algebraic $\partial_{\scalelog}\in\End(\K[\scalelog])$ to the analytic differential operator).

\begin{beispiel}
	The single Mellin transform $F(\reg)=\frac{1}{\reg(1-\reg)}$ in the propagator-type equation ($\powdep=-2$ as in example \ref{ex:DSE-propagator}) occurs in the first approximation to quantum electrodynamics and the Yukawa theory. The equation
\begin{equation*}
	\frac{\coupling}{G_{\scalelog}(\coupling)}
	=
	\partial_{-\scalelog} \left( 1-\partial_{-\scalelog} \right) G_{\scalelog} (\coupling)
	\urel{\eqref{eq:propagator-coupling-diff}}
	\widetilde{\toylog}(\coupling) \left(1 -2\coupling\partial_{\coupling} \right)
	\left[ 
	1	-	\widetilde{\toylog}(\coupling) \left(1-2\coupling\partial_{\coupling} \right)
		\right] G_{\scalelog}(\coupling),
\end{equation*}
evaluates at $\scalelog=0$ to the compact form
\begin{equation}
	\coupling
	=	\widetilde{\toylog}(\coupling) 
		-\widetilde{\toylog}(\coupling) (1-2\coupling\partial_{\coupling})\widetilde{\toylog}(\coupling).
\end{equation}
We stress how this equation determines the anomalous dimension non-perturbatively, in fact it can be expressed in terms of the complementary error function as analyzed in \cite{Yeats,Kreimer:ExactDSE}. Reference \cite{BaalenKreimerUminskyYeats:QED} is devoted to a detailed study of this type of equations and also solves the case $\powdep=1$ with the help of the Lambert $W$ function.
\end{beispiel}

\subsection{Variations of Mellin transforms}
Consider a change of the Mellin transform $F$ to a different $\widetilde{F}$ that keeps $c_{-1}$ fixed but is free to alter the other coefficients $\coeff{n}$ with $n\in\N_0$. Then by \eqref{eq:poly-coboundaries} the difference
\begin{equation}
	\dH\alpha 
	\defas P \circ \left[\widetilde{F} - F\right]_{-\partial_{\scalelog}}
	\in \HB[1](\K[\scalelog])
\end{equation}
is a Hochschild-1-coboundary and \eqref{eq:change-coboundary-equals-auto} shows that we can relate the two resulting renormalized Feynman rules $\phiR$ and $\widetilde{\phiR}$ by composition with a distinguished Hopf algebra automorphism:
	\begin{equation*}
		\widetilde{\phiR}
		\urel{\eqref{eq:phiR-universal}} \unimor{P \circ \widetilde{F}(-\partial_{\scalelog})}
		= \unimor{P \circ F(-\partial_{\scalelog}) + \dH\alpha}
		\urel{\eqref{eq:change-coboundary-equals-auto}}
		\unimor{P \circ F(-\partial_{\scalelog})} \circ \auto{\left[ \alpha\, \circ\, \unimor{P \circ F(-\partial_{\scalelog})} \right]}
		\urel{\eqref{eq:phiR-universal}}
		\phiR \circ \auto{\left[ \alpha\, \circ\, \phiR \right]}.
	\end{equation*}
	This accentuates that the Feynman rules coming from \eqref{eq:renormiert-def} obey even more structure than just the renormalization group \eqref{eq:toyexp}: The origin from the Mellin transform poses restrictions on the generator $\toylog$ of $\phiR$. For illustration consider how example \ref{ex:toylog-mellin-recursion} implies
	\begin{equation}
		\toylog\left( \tree{+-} \right)
		\cdot \left[ 2 \toylog\left( \tree{+++---} \right) - \toylog\left( \tree{++-+--} \right) \right]
		= 2 \coeff[2]{-1} \coeff[2]{0}
		= 2 \left[ \toylog\left( \tree{++--} \right) \right]^2.
	\end{equation}
\begin{beispiel}
	Assume that $\coeff{-1}=-1$, then we can relates $\phiR$ to the tree factorial character $\intrules=\unimor{\polyint}$: The difference in Mellin transforms is
	\begin{equation}
		\dH \alpha 
		= P \circ \sum_{n \in \N_0} \coeff{n} \partial_{-\x}^n \in \HB[1](\K[\x]),
		\quad\text{therefore}\quad
		\alpha
		\urel{\eqref{eq:poly-coboundaries}}
		\counit \circ \sum_{n\in\N_0} \coeff{n} \partial_{-\x}^n
		\in \K[\x]'
	\end{equation}
	such that we find $\alpha \circ \intrules (\f) = \alpha\left( \frac{\x^{\abs{\f}}}{\f!} \right) = (-1)^{\abs{\f}} \frac{\abs{\f}!}{\f!} \coeff{\abs{\f}}$. Now we can verify
\begin{align*}
	\phiR \left( \tree{+-} \right)
		&= \x
		 = \intrules \left( \tree{+-} \right)
		 \urel{\ref{ex:auto}}
		 \intrules \circ \auto{\alpha \circ \intrules} \left( \tree{+-} \right)
		,\qquad
	\phiR \left( \tree{++--} \right)
		 = \frac{\x^2}{2} + \coeff{0} \x
		 = \intrules \left\{ \tree{++--} + \toyform(1) \tree{+-} \right\}
		 \urel{\ref{ex:auto}}
		 \intrules \circ \auto{\alpha \circ \intrules} \left( \tree{++--} \right)
		 ,\\
	\phiR \left( \tree{+++---} \right)
		&= \frac{\x^3}{6} + \x^2 \coeff{0} + \x (\coeff[2]{0}-\coeff{1})
		= \intrules \left\{ \tree{+++---} + 2\coeff{0} \tree{++--} + \left[\coeff[2]{0}-\coeff{1} \right] \tree{+-} \right\}
		\urel{\ref{ex:auto}}
		\intrules \circ \auto{\alpha \circ \intrules} \left( \tree{+++---} \right)
		\quad\text{and}\\
	\phiR \left( \tree{++-+--} \right)
		&= \frac{\x^3}{3} + \coeff{0} \cdot \x^2 - 2\coeff{1} \cdot \x
		=\intrules \Big\{ \tree{++-+--} + \coeff{0} \tree{+-}\tree{+-} - 2\coeff{1} \tree{+-} \Big\}
		\urel{\ref{ex:auto}}
		\intrules \circ \auto{\alpha \circ \intrules} \left( \tree{++-+--} \right).
\end{align*}
\end{beispiel}
Having related different Feynman rules by such automorphisms, we may ask how the actual correlation functions are influenced by such a change. Here we can observe a kind of rigidity of the Dyson-Schwinger equation in
\begin{korollar}\label{satz:change-mellin-auto}
	The new correlation function $\phiR \circ X = \intrules \circ \widetilde{X}$ equals the original Feynman rules $\intrules$ applied to a modified perturbation series $\widetilde{X}(\coupling)$, fulfilling an equivalent Dyson-Schwinger equation that merely differs in the cocycles by coboundaries. By \eqref{eq:auto-leading-term} the leading logs coincide and explicitly
	\begin{align*}
		\widetilde{X} (\coupling)
		\defas \auto{\alpha_{\cdot} \circ \intrules} \circ X(\coupling)
		= \1 + 	\sum_{n\in\N} \coupling^n \left( B_n + \dH \alpha_n \right) \left( \widetilde{X}(\coupling)^{1+n\powdep} \right).
	\end{align*}
\end{korollar}
This follows directly through application of $\auto{\alpha \circ \intrules}$ to the original Dyson-Schwinger equation. Since we consider now many cocycles $B_n$, each of which corresponding to a different Mellin transform $F_n$, also the functionals $\alpha_n$ are now indexed by $n$.

\subsection{Minimal subtraction}
In section \ref{sec:running-coupling} we already understood precisely why the arbitrary choice of the renormalization point $\rp$ in the {\momscheme} does not influence the physical interpretation. Now we can find an equivalent relation for the minimal subtraction scheme in

\begin{korollar}
	Applying \eqref{eq:perturbation-convolution-char} to \eqref{eq:phiR-MSR} or \eqref{eq:MSR-coprod} expresses the correlation function $G_{{\scriptscriptstyle\mathrm{MS}}}$ of the $\Rms$-scheme in terms of $G$ in the {\momscheme} by a redefinition of the coupling constant and an overall factor:
	\begin{equation}
		G_{{\scriptscriptstyle\mathrm{MS}},\scalelog}(\coupling)
		= G_{{\scriptscriptstyle\mathrm{MS}},0}(\coupling)	\cdot
		G_{\scalelog} \Big( \coupling\cdot \left[ G_{{\scriptscriptstyle\mathrm{MS}},0}(\coupling) \right]^{\powdep} \Big).
	\end{equation}
\end{korollar}

\section{Extensions towards {\qft}}
\label{sec:extensions}

\subsection{Feynman graphs and Feynman integrals}
In the formulation of perturbative {\qft}, the Hopf algebra $H_R$ of rooted trees is replaced by the Hopf algebra $H_{FG}$ of Feynman graphs \cite{CK:RH1}.
Most importantly, it features insertion operators that act like Hochschild-1-cocycles on the relevant subspace of $H_{FG}$. Under the Feynman rules, these result in (divergent) sub integrals just as in \eqref{eq:unrenormiert-def} and the coproduct of $H_{FG}$ again mirrors the structure of subdivergences by definition.

The integrals are multi-dimensional and entail algebraic functions as integrands.
For massive theories, a \emph{Wick rotation} to Euclidean space-time disposes of all singularities in these integrands such that divergences only occur from the integrations at infinity.
These may be renormalized by suitable subtractions again, though one has to face two new issues: Multiple parameters and higher degrees of divergence.

Before addressing these we remark that dimensional regularization \cite{Collins} can be used to introduce a regulator $\reg \in \C \setminus\set{0}$ like in section \ref{sec:regulator} and assign a Laurent series $\toy$ in $\reg$ to every Feynman graph.
Then dimensional analysis reveals a scale dependence of the form
\begin{equation}
		\toy =
		\restrict{\toy}{s=\tilde{s}} \circ \gradAut_{-\reg\scalelog},
\end{equation}
where we employ the grading $Y$ by \emph{loop number} of $H_{FG}$ and $\scalelog=\ln\frac{s}{\tilde{s}}$ encodes the ratio of simultaneous rescaling of all dimensionful parameters.
Therefore the techniques of section \ref{sec:regulator} become available in {\qft} as shown for the first time in \cite{CK:RH2}.

\subsection{Multiple parameters}
Correlation functions of {\qft} typically depend on many variables, namely the masses of internal particles and the momenta of external particles. For illustration consider a logarithmic divergence with two parameters $(s,t)$: We can still renormalize
\begin{align}\label{eq:two-parameters}
	\phiR[(s,t)] =
		\int_0^\infty \left[ 
			\frac{x\ \dd x}{(x+s)(x+t)}
			-
			\frac{x\ \dd x}{(x+\tilde{s})(x+\tilde{t})}
		\right]
		= \frac{s\ln s - t\ln t}{t-s}
			-
			\frac{\tilde{s}\ln\tilde{s} - \tilde{t}\ln\tilde{t}}{\tilde{t}-\tilde{s}}
\end{align}
by a single subtraction at a reference point $(\tilde{s},\tilde{t})$ in the parameter space.
However, the function is no longer a plain logarithm in a single variable $\scalelog$.
In fact the dependence of correlation functions on the parameters becomes indeed extremely complicated and is only fully understood for the simplest (one-loop) Feynman graphs or slightly better in special situations like massless or supersymmetric theories, with on-shell conditions or in space-time dimensions different from four.

Crucially though the fundamental properties of a theory are described by asymptotic behavior, and the renormalization group still persists: If we rescale all parameters simultaneously by a factor $e^{\scalelog}$, then \eqref{eq:two-parameters} simplifies drastically to
$
	\phiR[(\tilde{s}\cdot e^{\scalelog}, \tilde{t} \cdot e^{\scalelog})]
	= -\scalelog.
$

Analogously to the single scale case we considered, the presence of subdivergences requests additional subtractions generating richer dependence on $\scalelog$, which nevertheless stays polynomial throughout.

Explicitly, encode all parameters as multiples of a distinguished \emph{scale} $s$ and dimensionless ratios $\theta \in \Theta$ called \emph{angles} and choose a renormalization point $(\tilde{s}, \tilde{\Theta})$ for the subtraction scheme $\momsch{(\tilde{s},\tilde{\Theta})}$ evaluating $(s,\Theta) \mapsto (\tilde{s},\tilde{\Theta})$. Then we can state (proof is provided in \cite{Panzer:Mellin} for Feynman integrals) the replacement for \eqref{eq:toyexp} as
\begin{satz}
	The renormalized Feynman rules
	$
		\phiR 
		= \restrict{\phiR}{\Theta=\tilde{\Theta}} \convolution \restrict{\phiR}{s=\tilde{s}}
	$
	factorize into the angle-dependent part $\restrict{\phiR}{s=\tilde{s}}$ (independent of $\scalelog$ and only a function of $\Theta$ and $\tilde{\Theta}$) and the scale-dependence $\restrict{\phiR}{\Theta=\tilde{\Theta}}$.
	The later depends only on $\scalelog=\frac{s}{\tilde{s}}$ (and the fixed renormalization point angles $\tilde{\Theta}$) and defines a morphism of Hopf algebras
	\begin{equation}
		\restrict{\phiR}{\Theta=\tilde{\Theta}}
		= \exp_{\convolution}\left( -\scalelog \period \right):
		\quad H \rightarrow \K[\scalelog].
	\end{equation}
	Its generator $\period \in \infchars{H}{\K}$, commonly called \emph{period}, is given by
$	\period %
	\defas \restrict{-\partial_{\scalelog}}{\scalelog=0} \restrict{\phiR}{\Theta=\tilde{\Theta}}$.
\end{satz}
We also recommend \cite{BrownKreimer:AnglesScales} for a different decomposition and detailed analysis of the angle- and scale dependence.

\subsection{Higher degrees of divergence}
So far we restricted ourselves to logarithmic divergences only. Recall from \eqref{eq:intro-()-def} that in this case though $\phi$ itself diverges, the derivative $\frac{\partial}{\partial s} \phi$ is convergent (by differentiating the integrand we obtain an integrable form).

In general one defines the \emph{superficial degree of divergence} $\sdd$ by simple power counting of the integrand $f$ such that $f(\zeta) \in \bigo{\zeta^{\sdd-n}}$, where $n$ counts the number of variables that we integrate over and the asymptotics are to be understood as all of these variables approaching $\infty$ jointly (for rational functions $f$, $\sdd$ is plainly the degree of the numerator minus the degree of the denominator, less the number of variables in the integral).

In this situation we find that any derivative $\frac{\partial^{\sdd+k}}{\partial s^{\sdd+k}} \phi$ for $k\in\N$ is convergent. 
Hence we can renormalize and keep these derivatives intact by subtracting a polynomial in $\K[s]_{\leq \sdd}$ of degree $\leq \sdd$.
In the logarithmic case, this freedom is precisely a single constant we parametrized by $\rp$ so far.

This renormalization scheme is common practice in {\qft} under the name BPHZ and the involved analytic estimates on the integrands necessary to prove the finiteness (as we did in theorem \ref{satz:finiteness} in the simple setup of definition \ref{def:unrenormiert}) have been worked out in \cite{Weinberg:HighEnergy,Zimmermann:Bogoliubov}. Variants exist for massless theories as well, while minimal subtraction (in connection with dimensional regularization) is particularly popular to handle gauge theories.

Subtractions of different polynomials in the external parameters are encoded into the Hopf algebra $H_{FG}$ by adjunction of auxiliary marked vertices and specification of \emph{external structures} \cite{CK:RH1}, also called \emph{form factor decomposition}.

As we apply different subtractions depending on the value of $\sdd$ for each individual graph, the algebraic formulation of section \ref{sec:renormalization} needs to be relaxed to encompass \emph{Rota-Baxter families} \cite{FardBondiaPatras:LieApproach}.
The \emph{exponential renormalization} introduced in \cite{FardPatras:ExponentialRenormalization} allows for even more general renormalization schemes and further explicitly captures the idea of order-by-order renormalization through counterterms common in physics.
It also shows the role of reparametrizations of the coupling constant we commented on in section \ref{sec:DSE} and we further recommend \cite{FardPatras:ExponentialRenormalization2} focussing on BPHZ.

We close by only briefly mentioning the increasing freedom in the Feynman rules coming along with growing degrees of divergence. If the divergences can attain arbitrarily high degrees, infinitely many subtraction terms are necessary in the renormalization process and thus generate as many constants (like our subtraction point $\rp$) to be fixed. Such a theory is unfortunately called \emph{unrenormalizable}: Though it \emph{is} renormalizable, it loses any predictive power due to the infinity of unknown constants.

Contrary, \emph{renormalizable} theories adhere to an upper bound on the degree of divergences that occur. Therefore only finitely many parameters (\emph{renormalization conditions}) have to be fixed through measurements, whereafter all other processes may in principle be predicted.

The divergences of scalar field theory together with their renormalization according to BPHZ is for example elaborated on in \cite{BrownKreimer:AnglesScales} in the parametric representation, though most textbooks contain at least a basic account of this theme.

\subsection{Overlapping divergences}
Historically, the possibility of overlapping (compared to nested and disjoint) subdivergences in Feynman graphs caused difficulties in proofs of renormalizability.
However the forest formula and its Hopf algebraic counterparts (the coproduct and the antipode) as for example formulated in \cite{CK:RH1} encompass overlapping divergences seamlessly, without the need of a special separate treatment.

The article \cite{Kreimer:Overlapping} explains how this is achieved and \cite{KrajewskiWulkenhaar:KreimersHopf} constructs an alternative Hopf algebra to address these structures and clarifies their equivalence.

Also note that the so-called \emph{multi-scale renormalization} avoids this issue altogether, as is explained in \cite{KrajewskiRivasseauTanasa:HopfMultiscale} with particular emphasis on its Hopf-algebraic formulation and the relations between these different Hopf algebras.

\subsection{Systems of Dyson-Schwinger equations}
Until now we only considered a single Dyson-Schwinger equation in section \ref{sec:DSE}, though {\qft}s typically involve different types of fields (like fermions, photons or scalars) and also a variety of couplings (\emph{vertices}). Each of those is represented by an (mostly) infinite series over Feynman graphs and these series may be inserted into each other in many ways.

These changes can be incorporated combinatorially by considering systems of Dyson-Schwinger equations as for example studied in detail in \cite{Foissy:Systems,Foissy:GeneralSystems}.

\section{Summary}
\label{sec:conclusion}

We reviewed renormalization of logarithmic ultraviolet divergences in the Hopf algebraic framework working with the rooted trees $H_R$.
In the {\momscheme} we arrived at the same renormalized Feynman rules $\phiR$ either by direct integration or with an analytic regulator being present.

After renormalization, the physical limit revealed a very special structure as being a Hopf algebra morphism \mbox{$\phiR\!: H_R \rightarrow \K[\x]$}.
This is the \emph{renormalization group property} and reduces the full character down to the linear terms $\toylog$ only. 

The minimal subtraction scheme does not allow for such a simple description: Though we could obtain the scale dependence, the constant terms in this scheme are not as easily understood.

All along the case of Feynman rules that can be described by the Mellin transform is special in that it gives simple explicit recursions for the renormalization process.

In section \ref{sec:DSE} studying the correlation functions, we understood the physical equivalence of two different renormalization schemes and the renormalization group.
We hinted at how these lead to possible non-perturbative formulations. 

Hochschild cohomology appeared ubiquitously in form of the universal property of rooted trees, governing most constructions we made.
Its importance lies not only in the concept and power of Dyson-Schwinger equations alone, but also in the induced automorphisms \eqref{eq:auto} of $H_R$ which help to understand variations of Feynman rules coming from different Mellin transforms.

\appendix
\section{The Hopf algebra of rooted trees}
\label{sec:H_R}

Lo\"{i}c Foissy studied rooted trees in depth and we only mention his thesis \cite{Foissy:PhD1,Foissy:PhD2} and in particular the article \cite{Foissy:FiniteComodules} which discloses the structure of $H_R$ as a free shuffle algebra. Here we restrict to introduce the notions relevant to understand this article and in particular elaborate on the universal property.

As an algebra, $H_R = \K[\trees]$ is free commutative\footnote{%
We consider \emph{unordered} trees $\scalebox{0.7}{$\tree{++-++---}$}=\scalebox{0.7}{$\tree{+++--+--}$}$ and forests $\tree{+-}\tree{++--}=\tree{++--}\tree{+-}$, sometimes called \emph{non-planar}.%
}
 generated by the \emph{rooted trees} $\trees$ and spanned by their disjoint unions (products) called \emph{rooted forests} $\forests$:
	\begin{equation*}
		\trees = 
			\set{\tree{+-}, \tree{++--}, \tree{+++---}, \tree{++-+--}, \tree{++++----}, \tree{+++-+---}, \tree{++-++---}, \tree{++-+-+--}, \ldots}
		,\quad
		\forests = 
				\set{\1} 
	\cup 	\trees
	\cup 	\set{\tree{+-}\tree{+-}, \tree{+-}\tree{+-}\tree{+-}, \tree{+-}\tree{++--}, \tree{+-}\tree{+-}\tree{+-}\tree{+-}, \tree{+-}\tree{+-}\tree{++--}, \tree{+-}\tree{++-+--}, \tree{+-}\tree{+++---}, \ldots}.
	\end{equation*}
Every $\f\in\forests$ is just the monomial $\f=\prod_{t\in\comps(\f)} t$ of its multiset of tree components $\comps(\f)$, while $\1$ denotes the empty forest. The number $\abs{\f}\defas \abs{V(\f)}$ of nodes $V(\f)$ induces the grading $ H_{R,n} = \lin \forests_n $ where $\forests_n \defas \setexp{\f\in\forests}{\abs{\f}  = n}$.
\begin{definition}
	The (linear) \emph{grafting operator} $B_+ \in \End(H_R)$ attaches all trees of a forest to a new root, for example
$	B_+ \left( \1 \right)	= \tree{+-} $,
$	B_+ \left( \tree{+-} \right) = \tree{++--}$ and
$	B_+ \left( \tree{+-}\tree{+-} \right)	= \tree{++-+--}$.
\end{definition}
$B_+$ is homogenous of degree one and restricts to a bijection $B_+\!:\ \forests \rightarrow \trees$. The coproduct $\Delta$ is defined to make $B_+$ a cocycle by requiring
\begin{equation}
	\Delta \circ B_+
	= B_+ \tp \1 + (\id \tp B_+) \circ \Delta.
	\label{eq:B_+-cocycle}
\end{equation}
\begin{lemma}\label{satz:B_+-cocycle}
	$B_+(\1) = \tree{+-} \neq 0$ implies that $0\neq [B_+] \in \HH[1](H_R)$ is non-trivial.
\end{lemma}
$H_R$ is characterized through the \emph{universal property} (theorem 2 in \cite{CK:NC}) of
\begin{satz}\label{satz:H_R-universal}
	For an algebra $\alg$ and $L \in \End(\alg)$ there exists a unique morphism $\unimor{L}\!: H_R \rightarrow \alg$ of unital algebras such that
	\begin{equation}
		\unimor{L} \circ B_+ = L \circ \unimor{L},
		\quad \text{equivalently} \quad
		\vcenter{\xymatrix{
		{H_R} \ar[r]^{\unimor{L}} \ar[d]_{B_+} & {\alg} \ar[d]^{L} \\
			{H_R} \ar[r]_{\unimor{L}} & {\alg}
		}}
		\quad \text{commutes.}
		\label{eq:H_R-universal}
	\end{equation}
	In case of a bialgebra $\alg$ and a cocycle $L \in \HZ[1](\alg)$, $\unimor{L}$ is a morphism of bialgebras and even of Hopf algebras when $\alg$ is Hopf.
\end{satz}
	Note that $\unimor{L} = e$ trivializes if $L\in\HB[1](\alg)$ is a coboundary.
\begin{beispiel}
	The action of $\unimor{L}$ is plainly replacement of $B_+$ by $L$:
\begin{equation*}
	\unimor{L} \left( \tree{++-+--} - 3\tree{+-} \right)
	= \unimor{L} \left\{ B_+ \left( {\left[B_+(\1) \right]}^2 \right) - 3 B_+ (\1) \right\}
	= L \left(  {\left[ L(\1_{\alg}) \right]}^2 \right)  - 3 L(\1_{\alg}).
\end{equation*}
\end{beispiel}
\begin{beispiel}\label{ex:tree-factorial}
	The cocycle $\polyint\in\HZ[1](\K[\x])$ of appendix \ref{sec:polynomials} induces the character
	\begin{equation}\label{eq:int-rules}
		\intrules \defas \unimor{\polyint}\in\chars{H_R}{\K[\x]}
		\quad\text{mapping any forest $\f\in\forests$ to}\quad
		\intrules (\f) 
		= \frac{\x^{\abs{\f}}}{\f!},
		\quad\text{using}
	\end{equation}
\end{beispiel}
\hide{
\begin{proof}
	The inductive proof (start at $f=\1$ trivial) supposes \eqref{eq:int-rules} to be true for all $f\in \forests_{\leq n}$. Then \eqref{eq:int-rules} also holds for any forest $f\in\forests_{n+1}$ with $\abs{\pi_0 (f)} > 1$ as
	\begin{equation*}
		  \intrules (f) 
		= \prod_{t \in \pi_0 (f)} \intrules (t) 
		\urel{\eqref{eq:int-rules}}
			\prod_{t \in \pi_0 (f)} \frac{\x^{\abs{t}}}{t!} 
		= \frac{\x^{\sum_{t\in \pi_0(f)} \abs{t}}}{\prod_{t \in \pi_0(f)} t!} 
		\urel{\eqref{eq:tree-factorial}}
			\frac{\x^{\abs{f}}}{f!},
	\end{equation*}
	exploiting $\abs{t} \leq n$ for any $t\in\pi_0(f)$ to use the induction hypothesis. It remains to consider a tree $t = B_+ (f)$ for some $f\in \forests_n$ in
	\begin{equation*}
		\intrules (t)
		= \intrules \circ B_+ (f) 
		\urel{\eqref{eq:H_R-universal}}
			\int_0 \circ \:\intrules (f) 
		\urel{\eqref{eq:int-rules}}
			\int_0^\x \frac{y^{\abs{f}}}{f!} \ \dd y
		= \frac{\x^{\abs{f}+1}}{(\abs{f}+1) \cdot f!} 
		= \frac{\x^{\abs{B_+ (f)}}}{\left( B_+ f \right)!}
		\urel{\eqref{eq:tree-factorial}}
			\frac{\x^{\abs{t}}}{t!}. \qedhere
	\end{equation*}
\end{proof}
}%
\begin{definition}\label{def:tree-factorial}
	The \emph{tree factorial} $(\cdot)!\in\chars{H_R}{\K}$ is for any $\f\in\forests$ given by
	\begin{equation}\label{eq:tree-factorial}
		\left[ B_+(\f) \right]! 
		= \f! \cdot \abs{B_+(\f)}
		,\quad\text{equivalently}\quad
		\f! 
		\urel{\footnotemark}
		\prod_{v\in V(\f)} \abs{\f_v}.
	\end{equation}
	\footnotetext{By $\f_v$ we denote the subtree of $\f$ rooted at the node $v\in V(\f)$.}
\end{definition}

\subsection{Automorphisms of \texorpdfstring{$H_R$}{H\_R}}

Applying the universal property to $H_R$ itself, adding coboundaries to $B_+$ leads to
\begin{definition}\label{def:auto}
	For any $\alpha \in H_R'$, theorem \ref{satz:H_R-universal} defines the Hopf algebra morphism
	\begin{equation}
		\auto{\alpha} \defas \unimor{B_+ + \dH\alpha}\!:\ H_R \rightarrow H_R
		\quad\text{such that}\quad
		\auto{\alpha} \circ B_+ = \left[ B_+ + \dH \alpha \right] \circ \auto{\alpha}.
		\label{eq:auto}
	\end{equation}
\end{definition}
\begin{beispiel}\label{ex:auto}
	The action on the simplest trees yields
\begin{align*}
	\auto{\alpha} \left( \tree{+-} \right) 
		&= \auto{\alpha} \circ B_+ (\1)
		 = B_+ (\1) + (\dH\alpha) (\1) 
		 = B_+ (\1)
		 = \tree{+-}, \\
	\auto{\alpha} \left( \tree{++--} \right)
		&= \auto{\alpha} \circ B_+ \left( \tree{+-} \right)
		 = \left( B_+ + \dH\alpha \right) \auto{\alpha} \left( \tree{+-} \right)
		 = \tree{++--} + \dH\alpha \left( \tree{+-} \right)
		 = \tree{++--} + \alpha(\1) \tree{+-}, \\
	\auto{\alpha} \left( \tree{+++---} \right)
		&
		 = \tree{+++---} + 2 \alpha(\1) \tree{++--} + \left\{ {\left[ \alpha(\1) \right]}^2 + \alpha\left( \tree{+-} \right) \right\} \tree{+-}
		\quad\text{and}\quad
	\auto{\alpha} \left( \tree{++-+--} \right)
		 = \tree{++-+--} + 2 \alpha \left( \tree{+-} \right) \tree{+-} + \alpha (\1) \tree{+-}\tree{+-}.
\end{align*}
\end{beispiel}
These morphisms capture how $\unimor{L}$ reacts to variation of $L$ by a coboundary in
\begin{satz}\label{satz:change-coboundary-equals-auto}
	Let $H$ denote a bialgebra, $L \in HZ^1_{\counit}(H)$ a 1-cocycle and further $\alpha \in H'$ a functional. Then for $\unimor{L},\unimor{L+\dH \alpha}\!:\ H_R \rightarrow H$ given through theorem \ref{satz:H_R-universal} and $\auto{\alpha \circ \unimor{L}}\!:\ H_R \rightarrow H_R$ from definition \ref{def:auto}, we have
	\begin{equation}
		\unimor{L + \dH \alpha} = \unimor{L} \circ \auto{\left[\alpha\, \circ\, \unimor{L}\right]},
		\quad\text{equivalently}\quad
		\vcenter{\xymatrix{
			{H_R} \ar[r]^{\unimor{L + \dH\alpha}} \ar[d]_{\auto{\alpha \circ \unimor{L}}}  & {H} \\
			{H_R} \ar[ur]_{\unimor{L}} &
			}}
		\quad\text{commutes.}
		\label{eq:change-coboundary-equals-auto}
	\end{equation}
\end{satz}

\begin{proof}
	As both sides of \eqref{eq:change-coboundary-equals-auto} are algebra morphisms, it suffices to prove it inductively for trees: Let it be true for a forest $\f\in\forests$, then it holds as well for the tree $ B_+(\f)$ by
	\begin{align*}
		&\unimor{L} \circ \auto{\left[\alpha \circ \unimor{L}\right]} \circ B_+ (\f)
		\wurel{\eqref{eq:H_R-universal}} \unimor{L} \circ \left[ B_+ + \dH \left( \alpha \circ \unimor{L} \right) \right] \circ \auto{\left[\alpha \circ \unimor{L}\right]} (\f) \\
		&= \Big\{ 
												L \circ \unimor{L}
		 					+ 
																				 (\dH \alpha) \circ \unimor{L}
			 \Big\} \circ \auto{\left[\alpha \circ \unimor{L}\right]} (\f) 
		= \left\{ L + \dH \alpha \right\} \circ 
		 			\underbrace{
							\unimor{L} \circ \auto{\left[\alpha \circ \unimor{L}\right]} (\f)
					}_{
							\unimor{L + \dH\alpha} (\f)
					}
		\urel{\eqref{eq:H_R-universal}} \unimor{L + \dH\alpha} \circ B_+ (\f).
	\end{align*}
	We used $(\dH\alpha) \circ \unimor{L} = \unimor{L} \circ \dH \left( \alpha \circ \unimor{L} \right)$, following from $\unimor{L}$ being a morphism of bialgebras.
\end{proof}

\begin{satz}\label{satz:auto}
	The map $\auto{\cdot}\!:\ H_R' \rightarrow \End_{\text{Hopf}}(H_R)$, taking values in the space of Hopf algebra endomorphisms of $H_R$, fulfills the following properties:
	\begin{enumerate}
		\item For any $\f\in\forests$ and $\alpha\in H_R'$, $\auto{\alpha}(\f)$ differs from $\f$ only by lower order forests:
			\begin{equation}
				\auto{\alpha} (\f) \in \f + H_R^{\abs{\f}-1}
				= \f + \bigoplus_{n=0}^{\abs{\f}-1} H_{R,n}.
				\label{eq:auto-leading-term}
			\end{equation}
		\item $\auto{\cdot}$ maps $H_R'$ into the Hopf algebra automorphisms $\Aut_{\text{Hopf}} (H_R)$.
		Its image is closed under composition, as for any \mbox{$\alpha,\beta \in H_R'$} we have
			\begin{equation}
				\auto{\alpha} \circ \auto{\beta} = \auto{\gamma}
				\quad\text{upon setting}\quad
				\gamma = \alpha + \beta \circ {\auto{\alpha}}^{-1}.
				\label{eq:auto-composition}
			\end{equation}
		\item The maps $\dH\!: H_R' \rightarrow HZ^1_{\counit}(H_R)$ and $\auto{\cdot}\!: H_R' \rightarrow \Aut_{\text{Hopf}}(H_R)$ are injective,
		thus the subgroup $\im \auto{\cdot} = \setexp{\auto{\alpha}}{\alpha \in H_R'} \subset \Aut_{\text{Hopf}}(H_R)$ induces a group structure on $H_R'$ with neutral element $0$ and group law $\autoconc$ given by
			\begin{equation}
				\alpha \autoconc \beta 
				\defas \auto{\cdot}^{-1} \left( \auto{\alpha} \circ \auto{\beta} \right)
				\urel{\eqref{eq:auto-composition}}
				\alpha + \beta \circ \auto{\alpha}^{-1}
				\quad\text{and}\quad
				\alpha^{\autoconc -1} = -\alpha \circ \auto{\alpha}.
				\label{eq:auto-group}
			\end{equation}
	\end{enumerate}
\end{satz}

\begin{proof}
	Statement \eqref{eq:auto-leading-term} is an immediate consequence of $\dH\alpha (H_R^n) \subseteq H_R^n$: Starting from $\auto{\alpha}\left( \tree{+-} \right) = \tree{+-}$, suppose inductively \eqref{eq:auto-leading-term} to hold for forests $\f, \f'\in \forests$. Then it obviously also holds for $\f \cdot \f'$ as well and even so for $B_+(\f)$ through
	\begin{equation*}
		\auto{\alpha} \circ B_+(\f)
		=\left[ B_+ + \dH\alpha \right] \circ \auto{\alpha} (\f)
		\subseteq \left[ B_+ + \dH\alpha \right] \left( \f + H_R^{\abs{\f}-1} \right)
		\subseteq B_+(\f) + H_R^{\abs{\f}}.
	\end{equation*}
	This already implies bijectivity of $\auto{\alpha}$, but applying \eqref{eq:change-coboundary-equals-auto} to $L=B_+ + \dH\alpha$ and $\auto{\tilde{\alpha}}$ for $\tilde{\alpha} \defas - \alpha \circ \auto{\alpha}$ shows $\id = \auto{\alpha} \circ \auto{\tilde{\alpha}}$ directly. We deduce bijectivity of all $\auto{\alpha}$ and thus $\auto{\alpha}\in\Aut_{\text{Hopf}}(H_R)$ with the inverse $\auto{\alpha}^{-1}=\auto{\tilde{\alpha}}$. Now \eqref{eq:auto-composition} follows from
	\begin{equation*}
		\auto{\left[\alpha+\beta \,\circ\, \auto{\alpha}^{-1}\right]}
		= \unimor{\left[B_+ + \dH \alpha \right] + \dH \left( \beta \,\circ\, \auto{\alpha}^{-1} \right)}
		\urel{\eqref{eq:change-coboundary-equals-auto}} \unimor{\left[B_+ + \dH \alpha\right]} \circ \auto{\big[\beta \ \circ\ \auto{\alpha}^{-1} \ \circ\ \unimor{\left(B_+ + \dH \alpha\right)}\big]}
		= \auto{\alpha} \circ \auto{\beta}.
	\end{equation*}
	Finally consider $\alpha,\beta \in H_R'$ with $\auto{\alpha} = \auto{\beta}$, then
	$
		0
		= (\auto{\alpha} - \auto{\beta}) \circ B_+
		= \dH \circ (\alpha - \beta) \circ \auto{\alpha}
	$
	reduces the injectivity of $\auto{\cdot}$ to that of $\dH$. But if $\dH\alpha = 0$, for all $n\in\N_0$
	\begin{equation*}
		0
		= \dH\alpha \left( {\tree{+-}}^{n+1} \right)
		= \sum_{i=0}^n \binom{n+1}{i} \alpha \left( {\tree{+-}}^i \right) {\tree{+-}}^{n+1-i}
		\quad\text{implies}\quad
		\alpha\left( {\tree{+-}}^n \right) = 0.
	\end{equation*}
	Given an arbitrary forest $\f\in\forests$ and $n\in\N$, the expression
	\begin{align*}
		0 
		&= \dH\alpha \left( {\tree{+-}}^n \f \right)
		= \f \underbrace{\alpha\left( {\tree{+-}}^n \right)}_0 + \sum_{\f} \sum_{i=0}^n \binom{n}{i}{\tree{+-}}^i \f' \alpha\left( {\tree{+-}}^{n-i} \f'' \right) 
		  + \sum_{i=1}^n \binom{n}{i} \bigg[ {\tree{+-}}^i \f \underbrace{\alpha\left( {\tree{+-}}^{n-i} \right)}_0 + {\tree{+-}}^i \alpha\left( \f {\tree{+-}}^{n-i} \right) \bigg]
	\end{align*}
	simplifies upon projection onto $\K \tree{+-}$ to
	$
		\alpha\left( \f {\tree{+-}}^{n-1} \right)
		= - \frac{1}{n}\sum_{\f:\ \f'=\tree{+-}} \alpha\left( {\tree{+-}}^n \f'' \right)
	$.
	Iterating this formula exhibits $\alpha(\f)$ as a scalar multiple of $\alpha\big( {\tree{+-}}^{\abs{\f}} \big)=0$ and proves $\alpha=0$.
\end{proof}

\subsection{Decorated rooted trees}
\label{sec:decorations}

Our observations generalize straight forwardly to the Hopf algebra $H_R(\decor)$ of rooted trees with decorations drawn from a set $\decor$. In this case, the universal property assigns to each $\decor$-indexed family $L_{\cdot}\!:\ \decor \rightarrow \End(\alg)$ the unique algebra morphism 
	\begin{equation*}
			\unimor{L_{\cdot}}\!: H_R(\decor) \rightarrow \alg
			\quad\text{such that}\quad
			\unimor{L_{\cdot}} \circ B_+^d = L_d \circ \unimor{L_{\cdot}}
			\quad\text{for any $d \in \decor$}.
	\end{equation*}
For cocycles $\im L_{\cdot} \subseteq HZ^1_{\counit}(\alg)$ this is a morphism of bialgebras and even of Hopf algebras (should $\alg$ be Hopf).
For a family $\alpha_{\cdot}\!\!:\ \decor \rightarrow H_R'(\decor)$ of functionals, setting $L_d^{\alpha_{\cdot}} \defas B_+^d + \dH \alpha_d $ yields an automorphism $\auto{\alpha_{\cdot}} \defas \unimor{L_{\cdot}^{\alpha_{\cdot}}}$ of the Hopf algebra $H_R(\decor)$.
Theorems \ref{satz:change-coboundary-equals-auto} and \ref{satz:auto} generalize in the obvious way.

\section{The Hopf algebra of polynomials}
\label{sec:polynomials}

\begin{lemma}
	Requiring $\Delta (x) = x\tp \1 + \1 \tp x$ induces a unique Hopf algebra structure on the polynomials ${\K}[x]$. It is graded by degree, connected, commutative and cocommutative with $\Delta \left( x^n \right) = \sum_{i=0}^n \binom{n}{i} x^i \tp x^{n-i}$ and the primitive elements are precisely $\Prim \left( {\K}[x] \right) = \K \cdot x $.
\end{lemma}
The \emph{integration operator} $\polyint\!: x^n \mapsto \frac{1}{n+1}x^{n+1}$ is a cocycle $\polyint \in \HZ[1]({\K}[x])$ as
	\begin{align*}
		\Delta \polyint \left( \frac{x^n}{n!} \right)
		&= \Delta \left( \frac{x^{n+1}}{(n+1)!} \right)
		= \sum_{k=0}^{n+1} \frac{x^k}{k!} \tp \frac{x^{n+1-k}}{(n+1-k)!} \\
		&= \frac{x^{n+1}}{(n+1)!} \tp \1 + \sum_{k=0}^n \frac{x^k}{k!} \tp \polyint \left(\frac{x^{n-k}}{(n-k)!} \right)
		= \left[ \polyint \tp \1 + \left(\id \tp \polyint \right) \circ \Delta \right] \left( \frac{x^n}{n!} \right),
	\end{align*}
and is not a coboundary since $\polyint 1 = x \neq 0$. In fact it generates the cohomology by
\begin{satz}
	$\HH[1] (\K[x]) = \K \cdot [ \polyint ]$ is one-dimensional as the 1-cocycles are
	\begin{equation}
		\HZ[1](\K[x]) 
		= \K \cdot \polyint 
			\ \oplus\ 
			\dH \left( \K[x]' \right)
		= \K \cdot \polyint 
			\ \oplus\ 
			\HB[1](\K[x]).
		\label{eq-polys-cycles}
	\end{equation}
\end{satz}

\begin{proof}
	For an arbitrary cocycle $L \in \HZ[1] (\K[x])$, lemma \ref{satz:cocycle-props} ensures $L(1) = x a_{-1}$ where $a_{-1} \defas \partial_0 L(1)$. Hence $\tilde{L} \defas L - a_{-1}\polyint \in \HZ[1]$ fulfills $\tilde{L}(1)=0$, so $L_0 \defas \tilde{L} \circ \int_0 \in \HZ[1]$ by
	\begin{align*}
		\Delta \circ L_0
		= (\id \tp \tilde{L}) \circ \Delta \circ \polyint + (\tilde{L} \tp 1) \circ \polyint
		= (\id \tp L_0) \circ \Delta + L_0 \tp 1 + \tilde{L}(1) \cdot \polyint.
	\end{align*}
	Repeating the argument inductively yields $a_n\defas\partial_0 L_n(1) = \partial_0 \circ L \circ \polyint^{n+1}(1)\in\K$ and $L_{n+1} \defas (L_n - a_n\polyint) \circ \polyint \in \HZ[1]$, so for any $n\in\N_0$ we may read off from
	\begin{align*}
		L\circ\polyint^n (1) 
		= a_{-1} \polyint^{n+1} (1) + \ldots + a_{n-2} \polyint^2 (1) + L_{n-1}(1)
		= a_{-1} \polyint \left( \polyint^n 1 \right) + \sum_{j=0}^{n-1} a_j \polyint^{n-j} (1)
	\end{align*}
	that indeed $L=a_{-1}\polyint + \dH\alpha$ for the functional $\alpha\defas\partial_0\circ L \circ \polyint$ with $\alpha (\frac{x^n}{n!}) = a_n$.
\end{proof}
\begin{lemma}\label{satz:poly-coboundaries}
	Up to subtraction
	$
		P
		= \dH\counit
		= \id - \ev_0\!:
		{\K}[\x] \twoheadrightarrow \ker\counit
		= \x{\K}[\x]
	$
	of the constant part, direct computation exhibits $\dH \alpha$ as the differential operator
	\begin{equation}\label{eq:poly-coboundaries}
		\dH \alpha
		= P \circ \sum_{n\in\N_0} \alpha \left( \tfrac{x^n}{n!} \right) \partial^n
		\in \End({\K}[\x])
		\quad\text{for any}\quad
		\alpha \in \K[\x]'.
	\end{equation}
\end{lemma}

\begin{lemma}\label{satz:poly-characters}
	As any character $\phi \in \chars{{\K}[x]}{\K}$ of $\K[x]$ is fixed by $\lambda\defas \phi(x)$, they are the group $\chars{{\K}[x]}{\K} = \setexp{\ev_{\lambda}}{\lambda \in \K}$ of evaluations (the counit $\counit = \ev_0$ equals the neutral element)
	\begin{equation}\label{eq:poly-characters}
		\K[x] \ni p(x)
		\mapsto \ev_{\lambda} (p) \defas p(\lambda)
		\quad \text{with the product} \quad
		\ev_a \convolution \ev_b = \ev_{a+b}.
	\end{equation}
\end{lemma}
\begin{proof}
	Note 
	$
		\left[ \ev_a \convolution \ev_b \right] \left( x^n \right)
		= { \left[ \ev_a(1) \cdot \ev_b(x) + \ev_a(x) \cdot \ev_b(1) \right] }^n
		= { ( b + a ) }^n
	$.
\end{proof}
\begin{lemma}\label{satz:poly-log}
	The isomorphism $(\K,+)\ni a \mapsto \ev_a \in \chars{\K[x]}{\K}$ of groups is generated by the functional $\partial_0 = \ev_0\circ\partial\in\infchars{\K[x]}{\K}$, meaning $\log_{\convolution} \ev_a = a \partial_0$ and $\ev_a = \exp_{\convolution} (a\partial_0)$.
\end{lemma}
\begin{proof}
	Expanding the exponential series reveals $\exp_{\convolution}(a\partial_0) (x^n) = a^n$ as a direct consequence of $\partial_0^{\convolution k} = \counit \circ \partial^{\convolution k} = \counit \circ \partial^k$:
\begin{equation*}
	\partial_0^{\convolution k} \left(\frac{x^n}{n!}\right)
	= \sum_{\mathclap{i_1 + \ldots + i_k=n}} \left(\partial_0 \tfrac{x^{i_1}}{i_1!}\right) \cdots \left(\partial_0 \tfrac{x^{i_k}}{i_k!} \right)
	= \sum_{\mathclap{i_1 + \ldots + i_k=n}} \delta_{1,i_1} \cdots \delta_{1,i_k}
	= \delta_{k,n}
	= \restrict{\partial^k}{0} \left( \frac{x^{n}}{n!} \right)
	.\qedhere
\end{equation*}
\end{proof}

\section{The Dynkin operator \texorpdfstring{$D=S\convolution Y$}{D=S*Y}}

We briefly present the crucial properties of $D$ which are employed in this text and further recommend in particular section 4 of \cite{FardBondiaPatras:LieApproach} as well as \cite{MenousPatras:LogarithmicDerivatives}.
\label{sec:S*Y}
\begin{definition}
	For some fixed connected graduation $Y$ of $H$, define operators
	$	D_Y	\defas S\convolution Y$
	and
	$	\pi_Y \defas Y^{-1} \circ D_Y = D_Y \circ Y^{-1}$.
\end{definition}

As $Y$ is invertible on $\ker\counit$, $\pi_Y$ is well-defined. Note that each of $\set{S, D_Y, \pi_Y}$ commutes with $Y$ and $Y^{-1}$.
\begin{proposition}\label{satz:S*Y-infchar}
	$D_Y, \pi_Y \in \infchars{H}{H}$ are infinitesimal characters with
	$	\K \cdot \1 \oplus (\ker\counit)^2 %
		= \ker D_Y %
		= \ker \pi_Y %
	$ and $D_Y - Y, \pi_Y - P$ map into $(\ker\counit)^2$.
\end{proposition}
\begin{proof}
	Clearly,	
	$%
		\K \cdot \1 \oplus (\ker\counit)^2 %
		\subseteq \ker D_Y%
	$
	is an immediate consequence of
	\begin{align*}
		D_Y \circ m
		&= m \circ (S \tp Y) \circ \Delta \circ m
		= m \circ (S \tp Y) \circ (m \tp m) \circ \tau_{2,3} \circ (\Delta \tp \Delta) \\
		&=
		m^3 \circ (S \tp S \tp Y \tp \id + S \tp S \tp \id \tp Y) \circ \tau_{1,2} \circ \tau_{2,3} \circ (\Delta \tp \Delta) \\
		&= m \circ \left[ 
				(S \convolution Y) \tp (S \convolution \id)
				+ (S \convolution \id) \tp (S \convolution Y)
			\right]
		= D_Y \tp \counit + \counit \tp D_Y.
	\end{align*}
	The reverse inclusion follows from $D_Y (x) = Yx + \sum_x (S x') (Y x'') = Yx \mod (\ker\counit)^2$ for $x\in\ker\counit$.
\end{proof}

\begin{korollar}\label{satz:S*Y-generates-H}
	$V_Y \defas \im D_Y = \im \pi_Y$ generates $H$ as an algebra and contains the primitive elements $\Prim(H)\subseteq V_Y$.
\end{korollar}

\begin{proof}
	$H_{n+1} \subseteq V_Y^{\cdot} \defas \sum_{n\geq 0} V_Y^n$ follows inductively from $H^n = \bigoplus_{i\leq n} H_i \subseteq V_Y^{\cdot}$ using $x \in \pi_Y (x) + m(H^n \tp H^n)$ for any $x \in H_{n+1}$. A primitive $p$ yields $D_Y(p) = S(p) \cdot 0 + S(\1) \cdot Y(p) = Y(p)$.
\end{proof}

\begin{proposition}
	$\pi_Y^2 = \pi_Y$ is a projection, hence its image complements the square of the augmentation ideal:
	$
		H
		= \K \cdot \1 \oplus V_Y \oplus (\ker\counit)^2,
		\ker\counit
		= V_Y \oplus (\ker\counit)^2
	$.
\end{proposition}

\begin{proof} Expand $D_Y^2 = m \circ (S \tp Y) \circ (m \tp m) \circ \tau_{2,3} \circ (\Delta \tp \Delta) \circ (S \tp Y) \circ \Delta$ to
	\begin{align*}
		D_Y^2
		&= m^3 \circ \left[ 
				S \tp S \tp Y \tp \id + S \tp S \tp \id \tp Y
				\right] \circ \tau_{1,2} \circ \tau_{2,3} \circ (\Delta \tp \Delta) \circ (S \tp Y) \circ \Delta \\
		&= m \circ \left[ 
					(S \convolution Y) \tp e + e \tp (S \convolution Y)
				\right] \circ (S \tp Y) \circ \Delta%
		= D_Y \circ Y. \qedhere
	\end{align*}
\end{proof}

\begin{proposition}
	From $ \Delta \circ D_Y %
	= \1 \tp D_Y + \left[ D_Y \tp m\circ(S \tp \id) \right] \circ \tau_{1,2} \circ \Delta^2$ %
	we deduce that $\K \cdot \1 \oplus V_Y$ is a right-coideal.
	Further, $\pi_Y$ and $D_Y$ map co-commutative elements to primitives as then $\Delta \circ \pi_Y = \1 \tp \pi_Y + \pi_Y \tp \1$.
\end{proposition}

\begin{proof}
	Apply $S \convolution \id = e = \1 \cdot \counit$ and $(\id \tp \counit) \circ \Delta = \id$ to
	\begin{align*}
		&\Delta \circ D_Y
		= (m \tp m) \circ \tau_{2,3} \circ \tau_{1,2} \circ \left[ 
				S \tp S \tp Y \tp \id + S \tp S \tp \id \tp Y
		\right] \circ \Delta^3 \\
		&= \Big\{ 
				(S\convolution Y) \tp \left[ m\circ(S\tp\id) \right]
			\Big\} \circ \tau_{1,2} \circ \Delta^2
		+ \Big\{ 
				(S\convolution \id) \tp \left[ m\circ(S\tp Y) \right]
			\Big\} \circ \tau_{1,2} \circ \Delta^2. \qedhere
	\end{align*}
\end{proof}

\begin{korollar}
	For cocommutative $H$, $\exp_{\convolution} \left( \pi_Y \right) \in \chars{H}{H}$ is a character that coincides with $\id$ on the generating subspace $\im \left( \pi_Y \right) = \Prim(H)$, hence
	\begin{equation}
		\exp_{\convolution} \left( \pi_Y \right) = \id
		,\qquad\text{equivalently}\qquad
		\log_{\convolution}  \left( \id \right) = \pi_Y.
	\end{equation}
\end{korollar}
In particular note that in this case $\pi_Y = \log_{\convolution} \left(\id \right)$ does not depend on the choice of grading $Y$.
Recalling that by the Milnor-Moore theorem for this case $H = S(\Prim(H))$ is just the symmetric algebra, $\pi_Y$ is nothing but the projection on $\Prim(H)$ corresponding to 
\begin{equation}\label{eq:cocommutative-symmetric}
	H 
	= \bigoplus_{n \geq 0} \Prim(H)^{\tp n}.
\end{equation}

But also in the non-cocommutative case we have
\begin{proposition}\label{prop:V_Ygenerator}
	$V_Y$ generates $H$ as a free algebra: $S(V_Y) = H$ (as algebras).
\end{proposition}

\begin{proof}
	The inclusion $V_Y \hookrightarrow H$ induces a unique morphism $\nu: S(V_Y) \twoheadrightarrow H$ of algebras which is surjective by \ref{satz:S*Y-generates-H}. For $n,m\in \N_0$ and $v_1,\ldots,v_m \in H$,
	\begin{equation}
		\pi_Y^{\convolution n}(v_1\cdots v_m)
		= \sum_{i_1+\ldots+i_m = n} \binom{n}{i_1 \cdots i_m} \pi_Y^{\convolution i_1}(v_1) \cdots \pi_Y^{\convolution i_m}(v_m)
		\tag{$\ast$}
	\end{equation}
	results from iteration of $\pi_Y \circ m = m \circ (e \tp \pi_Y + \pi_Y \tp e)$ and proves
	\begin{equation*}
		\pi_Y^{\convolution n} (V_Y^m)
		= 0
		\quad\text{for any}\quad 0\leq n < m
	\end{equation*}
	as in $(\ast)$ some $i_k$ must vanish and $\pi_Y^{\convolution 0} = e$ annihilates $V_Y \subset \ker\counit$. 
	In the case $n=m$ we find $\restrict{\pi_Y^{\convolution n}}{V_Y^n} = n! \cdot \restrict{\id}{V_Y^n}$ by $i_1=\ldots=i_n=1$.
	Therefore a finite sum $0=\sum_{n\geq 0} x_n \in H$ with $x_n \in V_Y^n$ implies $x_0 = \frac{1}{0!} \pi_Y^{\convolution 0} (x) = 0$, then $x_1 = \frac{1}{1!} \pi_Y^{\convolution 1} (x) = 0$ and hence iteratively $x_n = 0$ for any $n$. 

	Thus $H = \bigoplus_{n\geq 0} V_Y^n$ is a direct sum with $V_Y^n = \nu(S_n(V_Y))$ upon the decomposition $S(V_Y) = \bigoplus_{n \geq 0} S_n(V_Y)$ into the homogeneous polynomials $S_n(V_Y)$ of degree $n$. 
	Since $\pi_Y^{\tp n} \circ \Delta^{n-1} (v_1\cdots v_n)%
	= \sum_\sigma v_{\sigma(1)}\tp\ldots\tp v_{\sigma(n)}$ for $v_1,\ldots,v_n \in V_Y$ delivers an inverse to $\restrict{\nu}{S_n(V_Y)}$, $\nu$ is injective.
\end{proof}
	Hence all commutative connected graded Hopf algebras are free, and $V_Y\supseteq\Prim(H)$ delivers a generating coideal.
	Note that we defined the algebra of rooted trees as $S(\lin\trees)$, but $V_Y$ is a much more special generator as it contains all primitives and hence roughly speaking many elements with simple coproducts.
	Especially in the cocommutative case $V_Y$ delivers the simplest possible generator $V_Y=\Prim(H)$ and $S(V_Y)=H$ becomes an isomorphism of Hopf algebras.
\begin{korollar}
	The map $\tilde{R}: \chars{H}{\alg} \rightarrow \infchars{H}{\alg}, \varphi \mapsto \varphi^{\convolution -1} \convolution (\varphi \circ Y) = \varphi \circ D_Y$ defined in \cite{Manchon} is a bijection, since the character $\varphi$ is already determined by its restriction on $V_Y = \im D_Y$ through corollary \ref{satz:S*Y-generates-H}.
\end{korollar}
\begin{satz}[\emph{scattering formula} from \cite{CK:RH2} in the form of \cite{Manchon}]
	The inverse of $\tilde{R}$ is given by $\tilde{R}^{-1} (\beta) = \lim_{t \rightarrow \infty} e^{-t Z_0} e^{t(Z_0 + \beta)}$ for any $\beta \in \infchars{H}{\alg}$.
	This equation is to be understood in the Lie group associated to the semidirect sum $\infchars{H}{\alg} \rtimes \C \cdot Z_0$, adjoining the derivation $Z_0(\beta) \defas \beta \circ Y$.
\end{satz}
Another description of $\tilde{R}^{-1}$ can be given as follows (see also \cite{FardBondiaPatras:LieApproach,MenousPatras:LogarithmicDerivatives}):
Denote by $\Psi \in \End\left( \End(H) \right)$ the map $\End(H) \ni \alpha \mapsto \Psi(\alpha) \defas (\alpha \convolution D_Y) \circ Y^{-1}$ and exploit $\id = e + Y \circ Y^{-1}$ as well as $Y = \id \convolution S \convolution Y = \id \convolution D_Y$ repeatedly, then
\begin{align}
\begin{split}
	\id
	&= e + (\id \convolution D_Y) \circ Y^{-1}
	= e + \left\{ 
			\left[ 
				e + (\id \tp D_Y) \circ Y^{-1}
			\right]
			\convolution D_Y
			\right\} \circ Y^{-1} \\
	&= \cdots
	= \sum_{n\geq 0} \Psi^n(e).
\end{split}
\end{align}
Observe that $\Psi^n(e)$ vanishes on any element $x\in H$ of coradical degree less than $n$, wherefore this series is pointwise finite.
Hence given $\beta \defas \tilde{R}(\varphi)$ we can reconstruct
\begin{align}
\begin{split}
	\varphi
	= \varphi \circ \id
	= e 
	&	+ \beta \circ Y^{-1}
		+ \left[ (\beta \circ Y^{-1}) \convolution \beta \right] \circ Y^{-1} \\
	&	+ \left( 
				\left\{ \left[ (\beta \circ Y^{-1}) \convolution \beta \right] \circ Y^{-1} \right\} \convolution \beta
				\right) \circ Y^{-1}
		+ \ldots.
\end{split}
\label{eq:Rtilde-inverse}
\end{align}

\bibliographystyle{amsplain}
\bibliography{qft}

\providecommand{\bysame}{\leavevmode\hbox to3em{\hrulefill}\thinspace}
\providecommand{\MR}{\relax\ifhmode\unskip\space\fi MR }
\providecommand{\MRhref}[2]{%
  \href{http://www.ams.org/mathscinet-getitem?mr=#1}{#2}
}
\providecommand{\href}[2]{#2}
\begin{thebibliography}{10}

\bibitem{BergbauerKreimer}
C.~Bergbauer and D.~Kreimer, \emph{{Hopf} algebras in renormalization theory:
  {Locality} and {Dyson}-{Schwinger} equations from {Hochschild} cohomology},
  IRMA Lect. Math. Theor. Phys. \textbf{10} (2006), 133--164.

\bibitem{BKW:NextToLadder}
I.~Bierenbaum, D.~Kreimer, and S.~Weinzierl, \emph{The next-to-ladder
  approximation for {Dyson}-{Schwinger} equations}, Phys. Lett. B \textbf{646}
  (2007), no.~2--3, 129--133.

\bibitem{BlochKreimer:MixedHodge}
S.~Bloch and D.~Kreimer, \emph{{Mixed Hodge Structures and Renormalization in
  Physics}}, Commun. Number Theory Phys. \textbf{2} (2008), no.~4, 637--718.

\bibitem{BroadhurstKreimer:Auto}
D.~J. Broadhurst and D.~Kreimer, \emph{Renormalization automated by {Hopf}
  algebra}, J. Symb. Comput. \textbf{27} (1999), 581.

\bibitem{Kreimer:ExactDSE}
\bysame, \emph{Exact solutions of {Dyson}-{Schwinger} equations for iterated
  one-loop integrals and propagator-coupling duality}, Nucl. Phys. B
  \textbf{600} (2001), no.~2, 403--422.

\bibitem{BrownKreimer:AnglesScales}
F.~C.~S. Brown and D.~Kreimer, \emph{Angles, scales and parametric
  renormalization}, Lett. Math. Phys. \textbf{103} (2013), no.~9, 933--1007
  (English).

\bibitem{Collins}
J.~C. Collins, \emph{Renormalization}, Cambridge Monographs on Mathematical
  Physics, Cambridge University Press, 1984.

\bibitem{CK:NC}
A.~Connes and D.~Kreimer, \emph{Hopf algebras, renormalization and
  noncommutative geometry}, Commun. Math. Phys. \textbf{199} (1998), no.~1,
  203--242.

\bibitem{CK:RH1}
\bysame, \emph{Renormalization in {Quantum} {Field} {Theory} and the
  {Riemann}-{Hilbert} {Problem} {I}: {The} {Hopf} {Algebra} {Structure} of
  {Graphs} and the {Main} {Theorem}}, Commun. Math. Phys. \textbf{210} (2000),
  no.~1, 249--273.

\bibitem{CK:RH2}
\bysame, \emph{Renormalization in {Quantum} {Field} {Theory} and the
  {Riemann}-{Hilbert} {Problem} {II}: {The} {$\beta$}-{Function},
  {Diffeomorphisms} and the {Renormalization} {Group}}, Commun. Math. Phys.
  \textbf{216} (2001), no.~1, 215--241.

\bibitem{FardBondiaPatras:LieApproach}
K.~{Ebrahimi-Fard}, J.~M. {Gracia-Bond{\'{\i}}a}, and F.~{Patras}, \emph{{A Lie
  Theoretic Approach to Renormalization}}, Commun. Math. Phys. \textbf{276}
  (2007), no.~2, 519--549.

\bibitem{FardPatras:ExponentialRenormalization}
K.~{Ebrahimi-Fard} and F.~Patras, \emph{Exponential renormalization}, Annales
  Henri Poincar{\'e} \textbf{11} (2010), no.~5, 943--971 (English).

\bibitem{FardPatras:ExponentialRenormalization2}
\bysame, \emph{{Exponential} renormalisation. {II}. {Bogoliubov}'s
  {R}-operation and momentum subtraction schemes}, J. Math. Phys. \textbf{53}
  (2012), no.~8, 083505.

\bibitem{Foissy:FiniteComodules}
L.~Foissy, \emph{Finite dimensional comodules over the hopf algebra of rooted
  trees}, J. Algebra \textbf{255} (2002), no.~1, 89--120.

\bibitem{Foissy:PhD1}
\bysame, \emph{Les alg{\`e}bres de {Hopf} des arbres enracin{\'e}s
  d{\'e}cor{\'e}s, {I}}, Bulletin des Sciences Math{\'e}matiques \textbf{126}
  (2002), no.~3, 193--239.

\bibitem{Foissy:PhD2}
\bysame, \emph{Les alg{\`e}bres de {Hopf} des arbres enracin{\'e}s
  d{\'e}cor{\'e}s, {II}}, Bulletin des Sciences Math{\'e}matiques \textbf{126}
  (2002), no.~4, 249--288.

\bibitem{Foissy:DSE}
\bysame, \emph{Fa{\`a} di {Bruno} subalgebras of the {Hopf} algebra of planar
  trees from combinatorial {Dyson}-{Schwinger} equations}, Advances in
  Mathematics \textbf{218} (2008), no.~1, 136--162.

\bibitem{Foissy:Systems}
\bysame, \emph{Classification of systems of {Dyson}-{Schwinger} equations in
  the {Hopf} algebra of decorated rooted trees}, Advances in Mathematics
  \textbf{224} (2010), no.~5, 2094--2150.

\bibitem{Foissy:GeneralSystems}
\bysame, \emph{General {Dyson}-{Schwinger} equations and systems}, Commun.
  Math. Phys. \textbf{327} (2014), no.~1, 151--179 (English).

\bibitem{periods}
M.~Kontsevich and D.~Zagier, \emph{Periods}, {Mathematics} {Unlimited} - 2001
  and {Beyond} (B.~Engquist and W.~Schmid, eds.), Springer, 2001, pp.~771--808.

\bibitem{KrajewskiRivasseauTanasa:HopfMultiscale}
T.~Krajewski, V.~Rivasseau, and A.~Tanasa, \emph{Combinatorial {Hopf} algebraic
  description of the multiscale renormalization in quantum field theory}, ArXiv
  e-prints (2012), 26.

\bibitem{KrajewskiWulkenhaar:KreimersHopf}
T.~Krajewski and R.~Wulkenhaar, \emph{On {Kreimer}'s {Hopf} algebra structure
  of {Feynman} graphs}, Eur. Phys. J. C \textbf{7} (1999), no.~4, 697--708.

\bibitem{Kreimer:HopfAlgebraQFT}
D.~Kreimer, \emph{On the {Hopf} algebra structure of perturbative quantum field
  theories}, Adv. Theor. Math. Phys. \textbf{2} (1998), no.~2, 303--334.

\bibitem{Kreimer:ChenII}
\bysame, \emph{Chen's iterated integral represents the operator product
  expansion}, Adv. Theor. Math. Phys. \textbf{3} (1999), no.~3, 627--670.

\bibitem{Kreimer:Overlapping}
\bysame, \emph{{On Overlapping Divergences}}, Commun. Math. Phys. \textbf{204}
  (1999), no.~3, 669--689.

\bibitem{Factorization}
\bysame, \emph{Factorization in quantum field theory: an exercise in {Hopf}
  algebras and local singularities}, On Conformal Field Theories, Discrete
  Groups and Renormalization (P.~Cartier, P.~Moussa, B.~Julia, and P.~Vanhove,
  eds.), Frontiers in Number Theory, Physics, and Geometry, vol.~2, Springer
  Berlin Heidelberg, 2007, pp.~715--736.

\bibitem{Kreimer:linear}
\bysame, \emph{{{\'E}tude} for linear {Dyson}-{Schwinger} {Equations}}, Traces
  in number theory, geometry and quantum fields (Sergio Albeverio, Matilde
  Marcolli, Sylvie Paycha, and Jorge Plazas, eds.), Aspects of Mathematics, no.
  E 38, Vieweg Verlag, 2008, pp.~155--160.

\bibitem{Panzer:Mellin}
D.~Kreimer and E.~Panzer, \emph{{Renormalization} and {Mellin} transforms},
  Computer Algebra in Quantum Field Theory (C.~Schneider and J.~Bl{\"u}mlein,
  eds.), Texts \& Monographs in Symbolic Computation, vol. XII, Springer Wien,
  September 2013, pp.~195--223.

\bibitem{Kreimer:NonlinearDSE}
D.~Kreimer and K.~A. Yeats, \emph{An {\'e}tude in non-linear
  {Dyson}-{Schwinger} equations}, Nucl. Phys. B (Proc. Suppl.) \textbf{160}
  (2006), 116--121, Proceedings of the 8th DESY Workshop on Elementary Particle
  Theory.

\bibitem{Manchon}
D.~Manchon, \emph{Hopf algebras in renormalisation}, Handbook of Algebra
  (Michiel Hazewinkel, ed.), Handbook of Algebra, vol.~5, Elsevier
  North-Holland, 2008, pp.~365--427.

\bibitem{MenousPatras:LogarithmicDerivatives}
F.~Menous and F.~Patras, \emph{Logarithmic derivatives and generalized {Dynkin}
  operators}, J. Algebraic Combin. \textbf{38} (2013), no.~4, 91--913
  (English).

\bibitem{Panzer:Master}
E.~Panzer, \emph{{Hopf}-algebraic {Renormalization} of {Kreimer}'s toy model},
  Master's thesis, Humboldt-Universit{\"a}t zu Berlin, July 2011.

\bibitem{Sweedler}
M.~E. Sweedler, \emph{Hopf algebras}, Mathematics Lecture Note Series, W. A.
  Benjamin, Inc., New York, 1969.

\bibitem{BaalenKreimerUminskyYeats:QED}
G.~{van Baalen}, D.~Kreimer, D.~Uminsky, and K.~A. Yeats, \emph{The {QED}
  {$\beta$}-function from global solutions to {Dyson}-{Schwinger} equations},
  Ann. Phys. \textbf{324} (2009), no.~1, 205--219.

\bibitem{Weinberg:HighEnergy}
S.~Weinberg, \emph{High-energy behavior in quantum field theory}, Phys. Rev.
  \textbf{118} (1960), no.~3, 838--849.

\bibitem{Yeats}
K.~A. Yeats, \emph{{Rearranging Dyson-Schwinger Equations}}, Mem. Amer. Math.
  Soc. \textbf{211} (2011), no.~995, 1--82.

\bibitem{Zimmermann:Bogoliubov}
W.~Zimmermann, \emph{Convergence of {Bogoliubov}'s method of renormalization in
  momentum space}, Commun. Math. Phys. \textbf{15} (1969), no.~3, 208--234
  (English).

\end{thebibliography}

\end{document}